\documentclass[twocolumn,superscriptaddress,aps]{revtex4-2}
\usepackage{graphicx}
\usepackage{tikz}
\usepackage{amsmath}
\usepackage{amsthm}
\usepackage{amssymb}
\usepackage{latexsym}
\usepackage{array}
\usepackage{hyperref}
\usepackage{float}
\usepackage{amsfonts}
\usepackage{dsfont}
\usepackage{mathrsfs}
\usepackage{verbatim}
\usepackage{bbold}
\usepackage{lipsum}
\usepackage[normalem]{ulem}

\usepackage{bm}
\usepackage{times}

\usepackage{upgreek}

\usepackage{makecell}
\usepackage{adjustbox,lipsum}

\usepackage{algorithm}
\usepackage{algpseudocode}

\usepackage{color}

\newcommand{\bra}[1]{\left<#1\right|}
\newcommand{\ket}[1]{\left|#1\right>}
\newcommand{\abs}[1]{\bigl|#1\bigr|}

\newcommand{\braket}[2]{\left<{#1}|{#2}\right>}
\newcommand{\ketbra}[2]{\ket{#1}\!\!\bra{#2}}

\newtheorem{theorem}{Theorem}
\newtheorem{proposition}[theorem]{Proposition}
\newtheorem{lemma}[theorem]{Lemma}
\newtheorem{corollary}[theorem]{Corollary}
\newtheorem{definition}[theorem]{Definition}

\newtheorem{remark}[theorem]{Remark}

\newcommand{\Tr}{\operatorname{Tr}}
\newcommand{\Var}{\operatorname{Var}}

\newcommand{\Prob}{\operatorname{Pr}}
\newcommand{\dd}{\mathrm{d}}
\providecommand{\openone}{\mathds{1}}
\newcommand{\one}{\hat{\mathds{1}}}
\newcommand{\vone}{\bm{1}}
\newcommand{\pp}{\bm p}
\newcommand{\xx}{\bm x}
\newcommand{\cB}{c_B}
\newcommand{\CB}{C_B}
\newcommand{\DCB}{\Delta C_B}
\newcommand{\DB}{\mathcal D_B}
\newcommand{\QB}{\mathcal Q_B}
\newcommand{\UU}{\mathcal U}
\newcommand{\EE}{\mathcal E}
\newcommand{\Hperp}{\hat{H}_d}
\newcommand{\Dfour}{\mathfrak D_d}

\begin{document}

\title{Coherence-generating power deviation: Fluctuations and input selectivity beyond average coherence generation}

\author{Kyoungho~Cho}\email{khcho23@yonsei.ac.kr}
\affiliation{Institute for Convergence Research and Education in Advanced Technology, Yonsei University, Seoul 03722, Republic of Korea}
\affiliation{Department of Statistics and Data Science, Yonsei University, Seoul 03722, Republic of Korea}

\author{Jeongho~Bang}\email{jbang@yonsei.ac.kr}
\affiliation{Institute for Convergence Research and Education in Advanced Technology, Yonsei University, Seoul 03722, Republic of Korea}
\affiliation{Department of Quantum Information, Yonsei University, Incheon 21983, Republic of Korea}

\date{\today}

\begin{abstract}
Coherence-generating power (CGP) quantifies the average coherence produced by a quantum process from incoherent inputs, but it cannot reveal whether coherence generation is robust or selective. We introduce the coherence-generating power deviation (CGPD), the standard deviation of Hilbert--Schmidt coherence over the uniform simplex of incoherent states. Expressing unitary coherence generation as a quadratic response on population space, we derive exact first and second moments in arbitrary finite dimension. We prove that every nontrivial unitary coherence generator necessarily fluctuates and decompose CGPD into an unavoidable isotropic contribution and excess anisotropy encoding input selectivity. This distinction yields operational consequences for benchmarking, reliability guarantees, ensemble susceptibility, and four-copy measurements, and extends to unital quantum channels. Matched quantum-walk and quasiperiodic-transport examples show that topology and localization can leave CGP unchanged while substantially modifying CGPD, thereby revealing dynamical structure hidden from the mean.
\end{abstract}

\maketitle

\section{Introduction}\label{sec:introduction}

Quantum coherence is a basis-dependent manifestation of superposition and a resource for quantum information processing, metrology, thermodynamics, and symmetry-sensitive tasks~\cite{Baumgratz2014,Streltsov2017,WinterYang2016,Napoli2016,Lostaglio2015,MarvianSpekkens2016}. Besides quantifying the coherence contained in a quantum state, one may ask how effectively a dynamical map converts an incoherent input into a coherent output. This question has led to several notions of coherence-generating power (CGP)~\cite{ManiKarimipour2015,ZanardiStyliarisVenuti2017,ZanardiStyliarisVenuti2017b,StyliarisVenutiZanardi2018,ZhangMaChenFei2018}.

The probabilistic CGP introduced in Ref.~\cite{ZanardiStyliarisVenuti2017} averages the Hilbert--Schmidt coherence generated by a unitary over a uniform ensemble of incoherent states. Its analytical tractability, basis covariance, and direct detection protocol make it a useful dynamical resource measure. Yet, CGP is a first-moment quantity. It answers how much coherence is produced on average, but not whether this production is distributed uniformly over the incoherent simplex or concentrated on a restricted set of population profiles.

This limitation is physical rather than merely statistical. In experiments and algorithms, the diagonal input may be imperfectly prepared, inferred only probabilistically, or deliberately drawn from a data distribution. Two operations with identical mean coherence generation can therefore perform differently from run to run. One may be a reliable resource supplier whose output coherence is insensitive to the population profile; another may be a selective converter that strongly amplifies only particular population imbalances. The mean value alone identifies neither behavior.

Fluctuation measures have recently been used to expose input-state dependence hidden by mean entangling power~\cite{ChoBang2026}. The underlying motivation carries over (as described in Appendix~\ref{app:relation}), but the coherence problem has a different ensemble geometry: incoherent inputs form a single probability simplex equipped here with the uniform Dirichlet measure, and unitary coherence generation is constrained by unistochastic population transport rather than by the product-state manifold of entangling power. Its fluctuations are therefore organized by population imbalances in the zero-sum subspace. We regard the generated coherence as a random variable over this simplex and define its standard deviation as the coherence-generating power deviation (CGPD). The deviation is not introduced as an error bar appended to CGP. It is a dynamical diagnostic in its own right: the width of the coherence-generation landscape, and hence a measure of input-population sensitivity, selectivity, and reliability.

To retain the information discarded by the mean, we treat coherence generation as a landscape over the incoherent simplex. The height of this landscape at a population vector $\pp$ is the coherence generated from the corresponding diagonal state. CGP records the mean height and CGPD records the width. Thus, two processes with the same CGP can represent very different physical behaviors: one may generate a coherence almost uniformly across input populations, while another may generate the same average coherence by responding strongly only to selected population imbalances.

Building on the coherence-matrix formulation of probabilistic CGP~\cite{ZanardiStyliarisVenuti2017,ZanardiStyliarisVenuti2017b}, this landscape can be written as a quadratic response on population space. For a unitary operator $\hat{U}$, the transition probabilities $\hat{P}_{\hat{U}}$ are obtained from the absolute squares of the matrix elements of $\hat{U}$, and the corresponding real response matrix is $\hat{\Gamma}_{\hat{U}}=\one-\hat{P}_{\hat{U}}^{\mathsf T}\hat{P}_{\hat{U}}$. Here, and throughout the paper, $\one$ denotes the identity matrix on real population space. CGP retains only $\Tr\hat{\Gamma}_{\hat{U}}$, whereas the exact fourth-order simplex moments show that CGPD additionally depends on $\Tr(\hat{\Gamma}_{\hat{U}}^2)$ and $\left\|\operatorname{diag}\hat{\Gamma}_{\hat{U}}\right\|_2^2$. These invariants resolve how the response is distributed among population modes and oriented relative to the simplex.

The main structural consequence is that the nonzero unitary coherence generation can never be perfectly input independent. The maximally mixed population sits at the center of the simplex and produces no coherence under unitary evolution, so a positive mean necessarily requires the landscape to rise away from the center. We make this statement quantitative by proving a universal dimension-dependent lower bound on CGPD and by decomposing the deviation into two parts: an unavoidable isotropic contribution fixed by CGP, and an excess anisotropic contribution that measures genuine input selectivity. The equality case identifies isotropic generators, while deviations above this boundary diagnose population modes that are preferentially converted into coherence.

Beyond this geometric picture, CGPD has several operational meanings. It is the single-input noise scale in finite-sample CGP benchmarking, gives distribution-independent guarantees against underperformance for unknown inputs, and appears as the susceptibility of average coherence generation to an infinitesimal bias in the input ensemble. We also give a four-copy measurement representation for the second moment, so the deviation can be accessed without full process tomography. The same moment formulas and fluctuation geometry extend to unital quantum channels through a Gram matrix of coherent output components.

Finally, we show that the deviation is sensitive to physical structure that the mean misses. The continuous-time quantum walks can be matched to have exactly the same CGP while having different CGPD, giving an analytically solvable topology-resolved example. A clean tight-binding chain and a localized Aubry--Andr\'e--Harper chain~\cite{AubryAndre1980,BiddleDasSarma2009,Lahini2009} show the same effect in transport: at a common CGP, localization produces a broader coherence-generation distribution. These examples connect CGPD to quantum simulation and to quantum-walk or Hamiltonian feature maps used in graph and kernel learning~\cite{Rossi2015,Henry2021,Albrecht2023,Havlicek2019,SchuldKilloran2019}.

\section{Preliminaries and statistical formulation}\label{sec:preliminaries}

\subsection{Basis coherence and incoherent states}

Let $\mathcal{H} \cong \mathbb C^d$ and fix an orthonormal reference basis $B = \{ \ket{i} \}_{i=1}^d$. The dephasing map and the complementary off-diagonal projection are
\begin{eqnarray}
\DB(\hat{X}) = \sum_{i=1}^d\ketbra{i}{i}\hat{X}\ketbra{i}{i},
\quad
\QB = \mathcal{I} - \DB,
\label{eq:dephasing}
\end{eqnarray}
where $\hat{X}$ is an operator on $\mathcal{H}$ and $\mathcal{I}$ is the identity superoperator on the operator space. Throughout this paper, $\hat{\openone}$ denotes the identity operator on Hilbert space, $\one$ denotes the identity matrix on real population space, and $\vone$ denotes the all-ones vector.
The map $\DB$ is an orthogonal projection with respect to the Hilbert--Schmidt inner product, and the operator space decomposes into diagonal and off-diagonal sectors. A density operator is $B$-incoherent when $\DB(\hat{\rho})=\hat{\rho}$.

We employ the squared Hilbert--Schmidt coherence used in the probabilistic CGP construction~\cite{ZanardiStyliarisVenuti2017},
\begin{eqnarray}
\cB(\hat{\rho}) &=& \left\|\QB(\hat{\rho})\right\|_2^2 = \Tr(\hat{\rho}^2) - \Tr[\DB(\hat{\rho})^2] \nonumber\\
	&=&\sum_{i\neq j}\abs{\bra{i}\hat{\rho}\ket{j}}^2.
\label{eq:coherence}
\end{eqnarray}
Although this quantity is not monotonic under every class of free coherence operations, it is a natural squared distance from the diagonal subspace, is monotonic for the unital incoherent maps relevant here, and renders the fluctuation theory exactly solvable.

Every incoherent state is represented by a population vector in the probability simplex,
\begin{eqnarray}
\hat{\rho}_{\pp} &=& \sum_{j=1}^dp_j\ketbra{j}{j},\nonumber\\
\Delta_{d-1} &=& \bigl\{\pp: p_j \geq 0, \ \sum_j p_j = 1 \bigr\}.
\label{eq:incoherent_state}
\end{eqnarray}
The maximally mixed population is denoted by $\pp_* = \vone/d$. Here, $\Delta_{d-1}$ denotes the $(d-1)$-dimensional probability simplex embedded in $\mathbb{R}^d$, since the population vectors satisfy $\sum_i p_i = 1$.

\subsection{Uniform ensemble of incoherent inputs}

The incoherent ensemble is the uniform probability measure $\nu_d$ on $\Delta_{d-1}$. It can be generated operationally by drawing a Haar-random pure state $\ket{\psi}$ and dephasing it,
\begin{eqnarray}
\hat{\rho}_{\pp} &=& \DB(\ketbra{\psi}{\psi}), \nonumber\\
p_j &=& \abs{\braket{j}{\psi}}^2.
\label{eq:haar_dephase}
\end{eqnarray}
The resulting population vector follows the Dirichlet distribution $\operatorname{Dir}(1,\ldots,1)$~\cite{ZyczkowskiSommers2001}. Thus, the Haar construction does not privilege any corner or face of the incoherent simplex. Appendix~\ref{app:haar} gives the moment derivation and explains the equivalence between Haar dephasing and the flat simplex measure.

The first two population moments are
\begin{eqnarray}
\mathbb{E}[p_i] = \frac{1}{d}, 
\quad
\mathbb{E}[p_ip_j] = \frac{1+\delta_{ij}}{d(d+1)}.
\label{eq:dirichlet_second}
\end{eqnarray}
It is useful to center the population vector,
\begin{eqnarray}
\xx = \pp - \frac{\vone}{d},
\qquad
\vone^{\mathsf T}\xx=0.
\label{eq:centered_population}
\end{eqnarray}
The covariance is isotropic on the zero-sum subspace,
\begin{eqnarray}
\mathbb{E}[\xx\xx^{\mathsf T}] = \frac{1}{d(d+1)}\Hperp,
\quad
\Hperp = \one-\frac{\vone\vone^{\mathsf T}}{d}.
\label{eq:centered_covariance}
\end{eqnarray}
This isotropy is the geometric origin of the trace formula for CGP.

\subsection{CGP, CGPD, and the coherence distribution}

For the unitary channel $\UU_{\hat{U}}(\hat{X})=\hat{U}\hat{X}\hat{U}^\dagger$, define the generated-coherence random variable
\begin{eqnarray}
X_{B,\hat{U}}(\pp) = \cB\!\left(\hat{U}\hat{\rho}_{\pp}\hat{U}^\dagger\right).
\label{eq:random_variable}
\end{eqnarray}
The full probability law of $X_{B,\hat{U}}$ is the coherence-generation distribution of $\hat{U}$ relative to $B$.

\begin{definition}[CGP and CGPD]
The coherence-generating power and coherence-generating power deviation are
\begin{eqnarray}
\CB(\hat{U}) &=& \mathbb E_{\nu_d}[X_{B,\hat{U}}], \nonumber \\
\DCB(\hat{U}) &=& \sqrt{\Var_{\nu_d}[X_{B,\hat{U}}]}.
\label{eq:definitions}
\end{eqnarray}
\end{definition}

CGP is the average height of the coherence-generation landscape over the simplex, while CGPD is its width. The coefficient of variation $\DCB/\CB$ is meaningful when $\CB>0$, but it mixes an unavoidable dimension-dependent fluctuation with additional structural anisotropy. 

The moment-generating and cumulant-generating functions are given as
\begin{eqnarray}
Z_{\hat{U}}(\lambda) = \mathbb E_{\nu_d}[e^{\lambda X_{B,\hat{U}}}],
\quad
K_{\hat{U}}(\lambda) = \log Z_{\hat{U}}(\lambda).
\label{eq:cumulant_generator}
\end{eqnarray}
Then, $K_{\hat{U}}'(0)=\CB(\hat{U})$ and $K_{\hat{U}}''(0)=\DCB(\hat{U})^2$. The pair $(\CB, \DCB)$ is therefore the first nontrivial truncation of a hierarchy of coherence-generation cumulants.

\subsection{Coherence-generation profile and terminology}

The distribution itself is cleanly defined as the push-forward of the simplex measure. For every Borel set $\mathcal A\subset\mathbb R$, let
\begin{eqnarray}
\mu_{B,\hat{U}}(\mathcal A) = \nu_d\!\left(\{\pp\in\Delta_{d-1}: X_{B,\hat{U}}(\pp)\in\mathcal A\}\right),
\label{eq:pushforward}
\end{eqnarray}
and write $F_{B,\hat{U}}(x) = \mu_{B,\hat{U}}(({-}\infty,x])$ for the cumulative distribution function. CGP and CGPD are properties of this coherence-generation profile, rather than independent state functionals. Two dynamics may therefore have the same CGP while differing in width, skewness, tails, or multimodality.

This observation motivates precise terminology. We call CGPD the absolute input-population sensitivity: it measures the typical change in generated coherence across the incoherent simplex. A generator is isotropic when it treats population directions as evenly as allowed by its nonzero mean, and selective when the response is concentrated on particular population directions. Because the maximally mixed input always yields zero coherence under a unitary, no nontrivial generator can have zero absolute sensitivity. Accordingly, isotropic generation should not be identified with $\DCB=0$; it is characterized instead by saturation of the dimension-dependent lower bound, which will be introduced in Sec.~\ref{sec:structure}. 

Quantiles provide a complementary language. The lower $\alpha$-quantile
\begin{eqnarray}
q_\alpha(\hat{U})&=&\inf\{x:F_{B,\hat{U}}(x)\geq\alpha\}
\label{eq:quantile}
\end{eqnarray}
answers how much coherence is guaranteed for a prescribed fraction of incoherent inputs. CGPD does not determine every quantile, but it gives distribution-independent guarantees and is the first statistic beyond the mean that constrains them.

\section{Population-space geometry of coherence generation}\label{sec:geometry}

\subsection{Unistochastic transition matrix}

The absolute squares of a unitary define the unistochastic matrix
\begin{eqnarray}
(\hat{P}_{\hat{U}})_{ij} = \abs{\bra{i}\hat{U}\ket{j}}^2.
\label{eq:unistochastic}
\end{eqnarray}
It is doubly stochastic: $\hat{P}_{\hat{U}} \vone = \vone$ and $\hat{P}_{\hat{U}}^{\mathsf T} \vone = \vone$. When a diagonal input has population $\pp$, the output population is $\bm q=\hat{P}_{\hat{U}}\pp$.

\begin{proposition}[Coherence-generation matrix]
\label{prop:gamma}
Define
\begin{eqnarray}
\hat{\Gamma}_{\hat{U}} = \one-\hat{P}_{\hat{U}}^{\mathsf T}\hat{P}_{\hat{U}}.
\label{eq:gamma}
\end{eqnarray}
Then,
\begin{eqnarray}
X_{B,\hat{U}}(\pp) = \pp^{\mathsf T}\hat{\Gamma}_{\hat{U}}\pp =\xx^{\mathsf T}\hat{\Gamma}_{\hat{U}}\xx.
\label{eq:quadratic_form}
\end{eqnarray}
Moreover, $\hat{\Gamma}_{\hat{U}}$ is real symmetric, $0\leq\hat{\Gamma}_{\hat{U}}\leq\one$, and $\hat{\Gamma}_{\hat{U}}\vone=0$.
\end{proposition}

\begin{proof}---The diagonal of $\hat{U}\hat{\rho}_{\pp}\hat{U}^\dagger$ is $\hat{P}_{\hat{U}}\pp$. Unitary invariance of purity and Eq.~(\ref{eq:coherence}) give
\begin{eqnarray}
X_{B,\hat{U}}(\pp) &=& \pp^{\mathsf T}\pp - (\hat{P}_{\hat{U}}\pp)^{\mathsf T}(\hat{P}_{\hat{U}}\pp) \nonumber \\
	&=& \pp^{\mathsf T}(\one-\hat{P}_{\hat{U}}^{\mathsf T}\hat{P}_{\hat{U}})\pp.
\end{eqnarray}
Because $\hat{P}_{\hat{U}}$ is doubly stochastic, its Euclidean operator norm is at most one, so $\hat{\Gamma}_{\hat{U}}\geq0$ and $\hat{\Gamma}_{\hat{U}}\leq\one$. Finally, $\hat{P}_{\hat{U}}\vone=\vone$ gives $\hat{\Gamma}_{\hat{U}}\vone=0$, which also implies the centered form in Eq.~(\ref{eq:quadratic_form}).
\end{proof}

This representation gives a sharp interpretation. The centered vector $\xx$ contains the classical population imbalances of the input. The eigenvectors of $\hat{\Gamma}_{\hat{U}}$ on $\vone^\perp$ are population modes, and the corresponding eigenvalues measure how efficiently each mode is converted into off-diagonal coherence. A monomial unitary has $\hat{P}_{\hat{U}}$ equal to a permutation matrix and $\hat{\Gamma}_{\hat{U}}=0$. A perfectly mixing complex Hadamard unitary has $\hat{P}_{\hat{U}}=\vone\vone^{\mathsf T}/d$ and $\hat{\Gamma}_{\hat{U}}=\Hperp$.

\subsection{Mean strength versus directional selectivity}

Let the nonzero eigenvalues of $\hat{\Gamma}_{\hat{U}}$ on $\vone^\perp$ be $\lambda_1,\ldots,\lambda_{d-1}$. Then,
\begin{eqnarray}
X_{B,\hat{U}}(\pp) = \sum_{\alpha=1}^{d-1}\lambda_\alpha y_\alpha(\pp)^2,
\label{eq:mode_expansion}
\end{eqnarray}
where $y_\alpha$ is the projection of $\xx$ onto the $\alpha$th population mode. CGP averages this quadratic response over the isotropic covariance in Eq.~(\ref{eq:centered_covariance}) and therefore retains only $\sum_\alpha \lambda_\alpha$. CGPD is sensitive to how unevenly the response is distributed among modes and to how those modes are oriented relative to the simplex vertices.


\subsection{Column participation and correlation structure}

Define the column inverse participation ratios and the overlap invariant
\begin{eqnarray}
s_j = \sum_i(\hat{P}_{\hat{U}})_{ij}^2,
\quad
S = \sum_j s_j,
\quad
R = \left\| \hat{P}_{\hat{U}}^{\mathsf T}\hat{P}_{\hat{U}} \right\|_F^2.
\label{eq:ipr}
\end{eqnarray}
Then,
\begin{eqnarray}
\Tr\hat{\Gamma}_{\hat{U}} &=& d-S, \nonumber \\
\Tr(\hat{\Gamma}_{\hat{U}}^2) &=& d-2S+R, \nonumber \\
\left\| \operatorname{diag}\hat{\Gamma}_{\hat{U}} \right\|_2^2 &=& \sum_j(1-s_j)^2.
\label{eq:gamma_invariants}
\end{eqnarray}
The quantity $S = \sum_{ij}\abs{\bra{i}\hat{U}\ket{j}}^4$ is the aggregate statistic already appearing in CGP. The invariant $R$ contains correlations between different transition-probability columns, and the final term records column-resolved inhomogeneity. These additional structures enter CGPD but not CGP.

\section{Exact moment theory}\label{sec:moments}

\subsection{Haar--Dirichlet tensor moments}

For $n \geq 1$, define the $n$-copy moment state of the uniform incoherent ensemble,
\begin{eqnarray}
\hat{\Omega}_B^{(n)}&=&\mathbb E_{\nu_d}[\hat{\rho}_{\pp}^{\otimes n}].
\label{eq:omega_n}
\end{eqnarray}
The Haar integration gives
\begin{eqnarray}
\hat{\Omega}_B^{(n)} &=& \DB^{\otimes n}\!\left[\frac{1}{(d)_n}\sum_{\pi\in S_n}\hat{T}_\pi\right], \nonumber \\
	(d)_n &=& d(d+1)\cdots(d+n-1),
\label{eq:tensor_moment}
\end{eqnarray}
where $\hat{T}_\pi$ permutes tensor factors~\cite{CollinsSniady2006}. In scalar form,
\begin{eqnarray}
\mathbb E[p_{i_1}\cdots p_{i_n}]&=&\frac{1}{(d)_n}\sum_{\pi\in S_n}\prod_{a=1}^n\delta_{i_a,i_{\pi(a)}}.
\label{eq:dirichlet_general}
\end{eqnarray}
The fourth-order case is sufficient for CGPD. Appendix~\ref{app:haar} derives Eq.~(\ref{eq:tensor_moment}), while Appendix~\ref{app:fourth} classifies the surviving permutations in the second-moment contraction.

\subsection{Closed formulas for the first two moments}

\begin{theorem}[Exact coherence moments]
\label{thm:moments}
Let
\begin{eqnarray}
\Dfour = d(d+1)(d+2)(d+3).
\label{eq:Dfour}
\end{eqnarray}
For every $d$-dimensional unitary operator $\hat{U}$,
\begin{eqnarray}
M_1(\hat{U}) = \mathbb E[X_{B,\hat{U}}] = \frac{\Tr\hat{\Gamma}_{\hat{U}}}{d(d+1)},
\label{eq:m1}
\end{eqnarray}
and
\begin{eqnarray}
M_2(\hat{U}) &=& \mathbb E[X_{B,\hat{U}}^2] \nonumber \\
	&=&\frac{(\Tr\hat{\Gamma}_{\hat{U}})^2 + 2\Tr(\hat{\Gamma}_{\hat{U}}^2) + 6\sum_{j=1}^d(\hat{\Gamma}_{\hat{U}})_{jj}^2}{\Dfour}.
\label{eq:m2}
\end{eqnarray}
Consequently,
\begin{eqnarray}
\DCB(\hat{U})^2&=&\frac{(\Tr\hat{\Gamma}_{\hat{U}})^2+2\Tr(\hat{\Gamma}_{\hat{U}}^2) + 6\left\|\operatorname{diag}\hat{\Gamma}_{\hat{U}}\right\|_2^2}{\Dfour} \nonumber \\
&& \quad -\frac{(\Tr\hat{\Gamma}_{\hat{U}})^2}{d^2(d+1)^2}.
\label{eq:cgpd_closed}
\end{eqnarray}
\end{theorem}

\begin{proof}---Using Proposition~\ref{prop:gamma} and Eq.~(\ref{eq:dirichlet_second}),
\begin{eqnarray}
M_1(\hat{U}) = \sum_{ij}(\hat{\Gamma}_{\hat{U}})_{ij}\mathbb E[p_ip_j] = \frac{\vone^{\mathsf T}\hat{\Gamma}_{\hat{U}}\vone+\Tr\hat{\Gamma}_{\hat{U}}}{d(d+1)},
\end{eqnarray}
which gives Eq.~(\ref{eq:m1}) because $\hat{\Gamma}_{\hat{U}}\vone=0$. For $M_2$, insert the fourth-order instance of Eq.~(\ref{eq:dirichlet_general}) into the contraction of two quadratic forms. All permutations with a fixed point produce a row-sum factor and vanish. The three double transpositions produce $(\Tr\hat{\Gamma}_{\hat{U}})^2+2\Tr(\hat{\Gamma}_{\hat{U}}^2)$, and the six four-cycles produce $6\sum_j(\hat{\Gamma}_{\hat{U}})_{jj}^2$. This yields Eq.~(\ref{eq:m2}); subtracting $M_1^2$ gives Eq.~(\ref{eq:cgpd_closed}). The complete index contraction is given in Appendix~\ref{app:fourth}.
\end{proof}

The first-moment formula recovers the CGP expression,
\begin{eqnarray}
\CB(\hat{U}) = \frac{1}{d(d+1)}\left(d-\sum_{ij}|\bra{i}\hat{U}\ket{j}|^4\right).
\label{eq:known_cgp}
\end{eqnarray}
The second moment can be written entirely in terms of $S$, $R$, and the individual $s_j$ through Eq.~(\ref{eq:gamma_invariants}). Explicitly,
\begin{eqnarray}
M_2(\hat{U}) = \frac{(d-S)^2+2(d-2S+R) + 6\sum_j(1-s_j)^2}{\Dfour}.
\label{eq:m2_transition}
\end{eqnarray}
Thus, two unitaries with equal $S$ necessarily have equal CGP but can have different CGPD through $R$ or the dispersion of the $s_j$.

\subsection{Directional information under inversion}

CGP is invariant under $\hat{U}\mapsto \hat{U}^\dagger$ because $\hat{P}_{\hat{U}}^{\mathsf T}\hat{P}_{\hat{U}}$ and $\hat{P}_{\hat{U}}\hat{P}_{\hat{U}}^{\mathsf T}$ have the same trace. CGPD need not be invariant because its diagonal term distinguishes column and row participation profiles. Define
\begin{eqnarray}
r_i = \sum_j(\hat{P}_{\hat{U}})_{ij}^2,
\quad
s_j=\sum_i(\hat{P}_{\hat{U}})_{ij}^2.
\label{eq:row_col_ipr}
\end{eqnarray}
Then,
\begin{eqnarray}
&& \DCB(\hat{U})^2 - \DCB(\hat{U}^\dagger)^2 \nonumber \\
&& \qquad = \frac{6}{\Dfour} \left( \sum_j(1-s_j)^2 - \sum_i(1-r_i)^2 \right).
\label{eq:directional_difference}
\end{eqnarray}
The signed difference is a forward--inverse asymmetry of coherence selectivity. It vanishes for symmetric transition matrices, for complex Hadamard unitaries, and in dimension two, but not in general.

For example, the real unitary
\begin{eqnarray}
\hat{U}_0 &=& \frac12\left(
\begin{array}{rrrr}
1&-1&0&\sqrt2\\
1&1&\sqrt2&0\\
-1&-1&\sqrt2&0\\
-1&1&0&\sqrt2
\end{array}\right)
\label{eq:u0}
\end{eqnarray}
has
\begin{eqnarray}
\CB(\hat{U}_0) &=& \CB(\hat{U}_0^\dagger)=\frac18, \nonumber \\
\DCB(\hat{U}_0)^2 &=& \frac{59}{6720}, \nonumber \\
\DCB(\hat{U}_0^\dagger)^2 &=& \frac1{120}.
\label{eq:u0_values}
\end{eqnarray}
Hence, CGPD can retain directional information discarded by the average.

\section{Structural properties and fluctuation geometry}\label{sec:structure}

\subsection{Invariances and basis covariance}

A $B$-incoherent unitary operator is monomial: it permutes basis states and attaches phases.

\begin{proposition}[Symmetries]
\label{prop:symmetry}
Let $\hat{W}_L$ and $\hat{W}_R$ be $B$-incoherent unitaries. Then, the full distributions of $X_{B,\hat{U}}$, $X_{B,\hat{W}_L\hat{U}}$, and $X_{B,\hat{U}\hat{W}_R}$ coincide. If $\widetilde{B}=\{\hat{V}\ket{i}\}$, then
\begin{eqnarray}
C_{\widetilde B}(\hat{U}) &=& \CB(\hat{V}^\dagger \hat{U}\hat{V}), \nonumber \\
\Delta C_{\widetilde B}(\hat{U}) &=& \DCB(\hat{V}^\dagger \hat{U}\hat{V}).
\label{eq:basis_covariance}
\end{eqnarray}
\end{proposition}

\begin{proof}---Left multiplication by $\hat{W}_L$ permutes rows of $\hat{P}_{\hat{U}}$ and leaves $\hat{P}_{\hat{U}}^{\mathsf T}\hat{P}_{\hat{U}}$ unchanged. Right multiplication by $\hat{W}_R$ permutes columns, so $\hat{\Gamma}_{\hat{U}}$ is conjugated by a permutation matrix. The uniform simplex measure is invariant under the same coordinate permutation, proving equality of the full distributions. Basis covariance follows by expressing dephasing and matrix elements in the transformed basis.
\end{proof}

The symmetry result shows that CGPD is a property of the nontrivial coherence-generating content of $\hat{U}$, not of basis relabelings or phases.

\subsection{Exact isotropic-plus-selective decomposition}

The fluctuation admits an exact isotropic-plus-selective decomposition. Define $\gamma_d=2d/\sqrt{(d-1)(d+2)(d+3)}$ and let
\begin{eqnarray}
\bar\lambda_{\hat{U}} =\frac{\Tr\hat{\Gamma}_{\hat{U}}}{d-1}, \quad \hat{A}_{\hat{U}} = \hat{\Gamma}_{\hat{U}}-\bar\lambda_{\hat{U}}\Hperp.
\label{eq:anisotropy_matrix}
\end{eqnarray}
Then, $\hat{A}_{\hat{U}}\vone=0$ and $\Tr\hat{A}_{\hat{U}}=0$. The matrix $\hat{A}_{\hat{U}}$ is the traceless anisotropic component of the coherence response.

\begin{theorem}[Fluctuation decomposition]
\label{thm:decomposition}
For every unitary operator $\hat{U}$,
\begin{eqnarray}
\DCB(\hat{U})^2 = \gamma_d^2\CB(\hat{U})^2 + \Xi_B(\hat{U})^2,
\label{eq:variance_decomposition}
\end{eqnarray}
where
\begin{eqnarray}
\Xi_B(\hat{U})^2 = \frac{2\|\hat{A}_{\hat{U}}\|_F^2 + 6\left\|\operatorname{diag}\hat{A}_{\hat{U}}\right\|_2^2}{\Dfour} \geq 0.
\label{eq:excess_selectivity}
\end{eqnarray}
Moreover, $\Xi_B(\hat{U})=0$ if and only if $\hat{\Gamma}_{\hat{U}}$ is isotropic on $\vone^\perp$.
\end{theorem}

\begin{proof}---The orthogonality relations $\Tr(\Hperp \hat{A}_{\hat{U}})=\Tr\hat{A}_{\hat{U}}=0$ and $\sum_j(\hat{A}_{\hat{U}})_{jj}=0$ imply
\begin{eqnarray}
\Tr(\hat{\Gamma}_{\hat{U}}^2) &=& \frac{(\Tr\hat{\Gamma}_{\hat{U}})^2}{d-1} + \left\|\hat{A}_{\hat{U}}\right\|_F^2,\nonumber\\
\left\|\operatorname{diag}\hat{\Gamma}_{\hat{U}}\right\|_2^2 &=& \frac{(\Tr\hat{\Gamma}_{\hat{U}})^2}{d} + \left\|\operatorname{diag}\hat{A}_{\hat{U}}\right\|_2^2.
\end{eqnarray}
Insert these identities into Eq.~(\ref{eq:cgpd_closed}). The terms proportional to $(\Tr\hat{\Gamma}_{\hat{U}})^2$ reduce to $\gamma_d^2\CB(\hat{U})^2$, and the remaining terms give Eq.~(\ref{eq:excess_selectivity}). Vanishing of $\Xi_B$ is equivalent to $\hat{A}_{\hat{U}}=0$.
\end{proof}

\begin{remark}
Since $\Xi_B(\hat{U})^2 \geq 0$, the decomposition can be written as
\begin{eqnarray}
\DCB(\hat{U})^2-\gamma_d^2\CB(\hat{U})^2 = \Xi_B(\hat{U})^2 \geq 0.
\label{eq:bound_remainder}
\end{eqnarray}
It therefore immediately yields the universal lower bound stated below. Moreover, the bound is saturated if and only if $\Xi_B(\hat{U})=0$, equivalently, $\hat{A}_{\hat{U}}=0$ or $\hat{\Gamma}_{\hat{U}}=\bar\lambda_{\hat{U}}\Hperp$. When $\CB(\hat{U})>0$, one has $\bar\lambda_{\hat{U}}>0$, which gives the equality condition in Eq.~(\ref{eq:isotropic_gamma}).
\end{remark}

\subsection{Universal lower bound}

\begin{theorem}[Universal CGP--CGPD bound]
\label{thm:lower_bound}
For every unitary operator $\hat{U}$ on $\mathbb C^d$,
\begin{eqnarray}
\DCB(\hat{U}) \geq \gamma_d\CB(\hat{U}).
\label{eq:lower_bound}
\end{eqnarray}
For $\CB(\hat{U})>0$, equality holds if and only if
\begin{eqnarray}
\hat{\Gamma}_{\hat{U}} &=& \lambda\Hperp
\label{eq:isotropic_gamma}
\end{eqnarray}
for some $\lambda>0$.
\end{theorem}

\begin{proof}---Eq.~(\ref{eq:bound_remainder}) gives the inequality immediately. For $\CB(\hat{U})>0$, equality holds if and only if $\Xi_B(\hat{U})=0$, which is equivalent to $\hat{\Gamma}_{\hat{U}}=\bar\lambda_{\hat{U}}\Hperp$ with $\bar\lambda_{\hat{U}}>0$. The converse follows by substitution.
\end{proof}

The theorem has a direct physical meaning. The maximally mixed input $\pp_*$ always produces zero coherence under a unital unitary evolution. If the average is positive, the landscape must rise away from the center, so some fluctuation is unavoidable. The equality condition describes an isotropic generator that treats every zero-sum population direction equally.

Eq.~(\ref{eq:variance_decomposition}) is central to the interpretation of CGPD. The first term is an irreducible radial fluctuation fixed entirely by the mean and dimension. The second term is excess fluctuation caused by directional selectivity. For $\CB(\hat{U})>0$, define
\begin{eqnarray}
\mathcal S_B(\hat{U}) = \frac{\DCB(\hat{U})}{\gamma_d\CB(\hat{U})} = \sqrt{1+\frac{\Xi_B(\hat{U})^2}{\gamma_d^2\CB(\hat{U})^2}} \geq 1.
\label{eq:selectivity_factor}
\end{eqnarray}
The value $\mathcal{S}_B=1$ identifies an isotropic generator; larger values quantify excess selectivity at fixed average power. Its reciprocal can be read as a uniformity index.

\begin{figure}[t]
\centering
\includegraphics[width=0.46\textwidth]{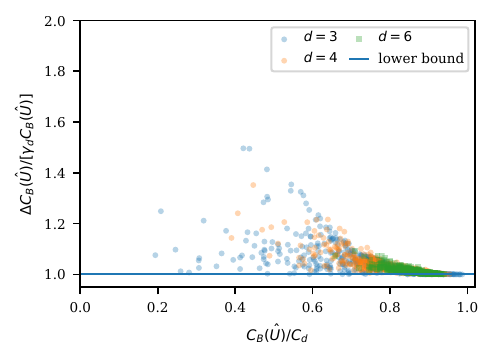}
\caption{Scaled relative deviation for Haar-random unitaries. The horizontal line is the exact lower bound in Eq.~(\ref{eq:lower_bound}). Points above the line have a nonzero anisotropic contribution $\Xi_B$. Highly mixing unitaries tend to approach the isotropic boundary.}
\label{fig:bound}
\end{figure}

Fig.~\ref{fig:bound} shows the allowed lower edge for Haar-random unitaries. The deviation above the edge is not sampling noise; it is the deterministic anisotropy term in Eq.~(\ref{eq:variance_decomposition}).

\subsection{Upper bound, faithfulness, and maximal generators}

\begin{proposition}[Allowed CGP--CGPD region]
\label{prop:upper_bound}
For every unitary operator $\hat{U}$,
\begin{eqnarray}
\gamma_d\CB(\hat{U}) &\leq& \DCB(\hat{U}) \nonumber \\
	&\leq& \sqrt{\CB(\hat{U})\left(\frac{d-1}{d} - \CB(\hat{U})\right)}.
\label{eq:two_sided_bound}
\end{eqnarray}
\end{proposition}

\begin{proof}---The lower inequality is Theorem~\ref{thm:lower_bound}. For any density operator, $0\leq\cB(\hat{\rho}) \leq (d-1)/d$. Hence, $0 \leq X_{B,\hat{U}} \leq M_d=(d-1)/d$ and $X_{B,\hat{U}}^2 \leq M_dX_{B,\hat{U}}$. Taking expectations gives $\Var{(X_{B,\hat{U}})}\leq\CB(\hat{U})[M_d-\CB(\hat{U})]$.
\end{proof}

\begin{corollary}[Faithfulness for unitaries]
\label{cor:faithful}
For unitary dynamics,
\begin{eqnarray}
\DCB(\hat{U})=0 &\Longleftrightarrow& \CB(\hat{U})=0 \nonumber \\
	&\Longleftrightarrow& \hat{U}\ \hbox{is $B$-incoherent}.
\label{eq:faithfulness}
\end{eqnarray}
\end{corollary}

\begin{proof}---If $\DCB(\hat{U})=0$, Eq.~(\ref{eq:variance_decomposition}) gives $\gamma_d^2\CB(\hat{U})^2+\Xi_B(\hat{U})^2=0$. Since both terms are nonnegative, $\CB(\hat{U})=0$ and $\Xi_B(\hat{U})=0$. Conversely, if $\CB(\hat{U})=0$, Eq.~(\ref{eq:m1}) gives $\Tr\hat{\Gamma}_{\hat{U}}=0$, and positivity of $\hat{\Gamma}_{\hat{U}}$ then implies $\hat{\Gamma}_{\hat{U}}=0$. Thus, $\hat{\Gamma}_{\hat{U}}$ is trivially isotropic on $\vone^\perp$, $\Xi_B(\hat{U})=0$, and Eq.~(\ref{eq:variance_decomposition}) gives $\DCB(\hat{U})=0$. This proves the first equivalence. Moreover, whenever $\CB(\hat{U})=0$, the same argument gives $\hat{\Gamma}_{\hat{U}}=0$. Hence, $\hat{P}_{\hat{U}}^{\mathsf T}\hat{P}_{\hat{U}}=\one$. A nonnegative doubly stochastic matrix with orthonormal columns must be a permutation matrix, so $\hat{U}$ is monomial. The converse is immediate.
\end{proof}

Maximal CGP is attained when $\abs{\bra{i}\hat{U}\ket{j}}^2=1/d$, namely when $\hat{U}$ maps the reference basis to a mutually unbiased basis~\cite{ZanardiStyliarisVenuti2017}. Such the unitaries are also minimally fluctuating among all operations with the same mean.

\begin{corollary}[Complex Hadamard generators]
\label{cor:mub}
If $\abs{\bra{i}\hat{U}\ket{j}}^2=1/d$ for all $i,j$, then
\begin{eqnarray}
\CB(\hat{U}) &=& \frac{d-1}{d(d+1)}, \nonumber \\
\DCB(\hat{U}) &=& \frac{2\sqrt{d-1}}{(d+1)\sqrt{(d+2)(d+3)}}, \nonumber \\
\mathcal S_B(\hat{U}) &=& 1.
\label{eq:mub_values}
\end{eqnarray}
\end{corollary}

\begin{proof}---For a complex Hadamard unitary, $\hat{P}_{\hat{U}}=\vone\vone^{\mathsf T}/d$ and $\hat{\Gamma}_{\hat{U}}=\Hperp$. Eqs.~(\ref{eq:m1}), (\ref{eq:cgpd_closed}), and (\ref{eq:selectivity_factor}) give the stated values.
\end{proof}

The relative deviation of a maximal generator scales as
\begin{eqnarray}
\frac{\DCB(\hat{U})}{\CB(\hat{U})} = \gamma_d \sim \frac{2}{\sqrt d}.
\label{eq:mub_scaling}
\end{eqnarray}
Thus, the high-dimensional complex Hadamard dynamics combines maximal average generation with decreasing relative fluctuation.

\subsection{Why qubits are exceptional}

\begin{corollary}[Qubit relation]
\label{cor:qubit}
For every qubit unitary operator,
\begin{eqnarray}
\DCB(\hat{U}) = \frac{2}{\sqrt5}\CB(\hat{U}).
\label{eq:qubit_relation}
\end{eqnarray}
\end{corollary}

\begin{proof}---Every $2 \times 2$ doubly stochastic matrix has the form
\begin{eqnarray}
\hat{P}_{\hat{U}} = \left(\begin{array}{cc} a & 1-a \\ 1-a & a \end{array}\right),
\end{eqnarray}
so $\hat{\Gamma}_{\hat{U}}$ is proportional to $\Hperp$. The equality condition in Theorem~\ref{thm:lower_bound} then yields Eq.~(\ref{eq:qubit_relation}).
\end{proof}

Up to phases, a qubit unitary operator is parameterized by a mixing angle $\theta$, giving
\begin{eqnarray}
\CB(\hat{U}) &=& \frac16\sin^2(2\theta), \nonumber \\
\DCB(\hat{U}) &=& \frac{1}{3\sqrt5}\sin^2(2\theta).
\label{eq:qubit_explicit}
\end{eqnarray}
The full distribution can also be obtained. If $X_{\max} = 3\CB(\hat{U})$, then
\begin{eqnarray}
f_{\hat{U}}(x) = \frac{1}{2\sqrt{X_{\max}x}}, \quad 0<x<X_{\max}.
\label{eq:qubit_density}
\end{eqnarray}
The zero-sum population space is one dimensional, so no directional anisotropy exists and $\Xi_B(\hat{U})=0$ for all qubit unitaries. Independent information beyond CGP first appears for $d \geq 3$. Appendix~\ref{app:lowdim} gives the low-dimensional derivation of this distribution and the four-level invariant check used below.

\subsection{Parallel composition}

For architectures composed of independent subsystems, the coherence-generation matrix obeys a simple product law.

\begin{proposition}[Tensor-product law]
\label{prop:tensor}
Let $\hat{U}$ and $\hat{V}$ act on dimensions $d_A$ and $d_B$. Then,
\begin{eqnarray}
\hat{\Gamma}_{\hat{U}\otimes\hat{V}} &=& \hat{\Gamma}_{\hat{U}} \otimes \one + \one\otimes\hat{\Gamma}_{\hat{V}} - \hat{\Gamma}_{\hat{U}} \otimes \hat{\Gamma}_{\hat{V}}.
\label{eq:tensor_gamma}
\end{eqnarray}
Writing $t_{\hat{U}} = \Tr\hat{\Gamma}_{\hat{U}}$ and $t_{\hat{V}} = \Tr\hat{\Gamma}_{\hat{V}}$,
\begin{eqnarray}
C_B(\hat{U}\otimes\hat{V}) = \frac{d_Bt_{\hat{U}} + d_At_{\hat{V}} - t_{\hat{U}}t_{\hat{V}}}
{d_Ad_B(d_Ad_B+1)}.
\label{eq:tensor_cgp}
\end{eqnarray}
\end{proposition}

\begin{proof}---Since $\hat{P}_{\hat{U}\otimes\hat{V}}=\hat{P}_{\hat{U}}\otimes \hat{P}_{\hat{V}}$,
\begin{eqnarray}
\hat{\Gamma}_{\hat{U} \otimes \hat{V}} = \one-(\one-\hat{\Gamma}_{\hat{U}}) \otimes (\one-\hat{\Gamma}_{\hat{V}}),
\end{eqnarray}
which expands to Eq.~(\ref{eq:tensor_gamma}). Taking the trace and using Eq.~(\ref{eq:m1}) gives Eq.~(\ref{eq:tensor_cgp}).
\end{proof}

Eq.~(\ref{eq:tensor_gamma}) shows that CGP is not simply additive: the final term removes the response already generated in both factors. The complete second-order composition law can also be stated using only three scalar invariants per subsystem.

\begin{proposition}[Second-order composition data]
Let $m=d_A$, $n=d_B$, and for $\hat{G}_A=\hat{\Gamma}_{\hat{U}}$ and $\hat{G}_B=\hat{\Gamma}_{\hat{V}}$, define
\begin{eqnarray}
&& t_A=\Tr\hat{G}_A, ~~ q_A=\Tr(\hat{G}_A^2), ~~ h_A=\sum_i(\hat{G}_A)_{ii}^2, \nonumber\\
&& t_B=\Tr\hat{G}_B, ~~ q_B=\Tr(\hat{G}_B^2), ~~ h_B=\sum_j(\hat{G}_B)_{jj}^2.
\label{eq:local_invariants}
\end{eqnarray}
For $\hat{G}_{AB}=\hat{\Gamma}_{\hat{U}\otimes\hat{V}}$, the corresponding invariants are
\begin{eqnarray}
t_{AB} &=& mn - (m - t_A)(n - t_B), \nonumber \\
q_{AB} &=& mn - 2(m - t_A)(n - t_B) \nonumber \\
	&& \quad +(m - 2t_A + q_A)(n - 2t_B + q_B), \nonumber \\
h_{AB} &=& n h_A + m h_B + h_A h_B + 2t_A t_B \nonumber \\
	&& \quad -2h_A t_B - 2t_A h_B.
\label{eq:composite_invariants}
\end{eqnarray}
Substitution of $(t_{AB},q_{AB},h_{AB})$ into Theorem~\ref{thm:moments} gives the exact CGPD of the product unitary.
\end{proposition}

\begin{proof}---Write $\hat{X}_A=\one-\hat{G}_A=\hat{P}_{\hat{U}}^{\mathsf T}\hat{P}_{\hat{U}}$ and $\hat{X}_B=\one-\hat{G}_B=\hat{P}_{\hat{V}}^{\mathsf T}\hat{P}_{\hat{V}}$. Eq.~(\ref{eq:tensor_gamma}) is $\hat{G}_{AB}=\one-\hat{X}_A\otimes \hat{X}_B$. Therefore,
\begin{eqnarray}
\Tr\hat{G}_{AB} &=& m n - \Tr(\hat{X}_A) \Tr(\hat{X}_B), \nonumber \\
\Tr(\hat{G}_{AB}^2) &=& m n - 2\Tr(\hat{X}_A) \Tr(\hat{X}_B) \nonumber\\
	&& \quad + \Tr(\hat{X}_A^2)\Tr(\hat{X}_B^2).
\end{eqnarray}
Using $\Tr(\hat{X}_A)=m-t_A$ and $\Tr(\hat{X}_A^2)=m-2t_A+q_A$, and similarly for $B$, gives the first two lines of Eq.~(\ref{eq:composite_invariants}). On the product basis,
\begin{eqnarray}
(\hat{G}_{AB})_{i\alpha,i\alpha} = a_i + b_\alpha - a_i b_\alpha,
\end{eqnarray}
where $a_i=(\hat{G}_A)_{ii}$ and $b_\alpha=(\hat{G}_B)_{\alpha\alpha}$. Expanding the square and summing over $i,\alpha$ gives the expression for $h_{AB}$.
\end{proof}

This result is useful for layered or parallel architectures: average strength and second-order selectivity can be propagated without constructing the full coherence distribution.

\section{Operational meaning and experimental accessibility}\label{sec:operational}

\subsection{Reliability for unknown inputs}

Suppose that the diagonal input is encountered randomly according to $\nu_d$. The probability that the generated coherence underperforms its mean by at least $\epsilon>0$ is bounded by Cantelli's inequality~\cite{Feller1971},
\begin{eqnarray}
\Prob\!\left[X_{B,\hat{U}} \leq \CB(\hat{U}) - \epsilon\right] \leq \frac{\DCB(\hat{U})^2}{\DCB(\hat{U})^2+\epsilon^2}.
\label{eq:cantelli}
\end{eqnarray}
At fixed CGP, smaller CGPD gives a stronger reliability guarantee. Conversely, larger CGPD signals a selective generator whose favorable inputs may be useful when the population profiles can be controlled or optimized.

A two-sided concentration statement follows from Chebyshev's inequality,
\begin{eqnarray}
\Prob\!\left[\abs{X_{B,\hat{U}} - \CB(\hat{U})} \geq \epsilon\right] \leq \frac{\DCB(\hat{U})^2}{\epsilon^2}.
\label{eq:chebyshev}
\end{eqnarray}
These inequalities are distribution agnostic. Solving Eq.~(\ref{eq:cantelli}) for a target failure probability $0<\delta<1$ gives the certified lower performance
\begin{eqnarray}
&& a_\delta = \sqrt{\frac{1-\delta}{\delta}}, \nonumber \\
&& L_\delta(\hat{U}) = \max\{0,\CB(\hat{U}) - a_\delta\DCB(\hat{U})\}, \nonumber \\
&& \Prob[X_{B,\hat{U}} \geq L_\delta(\hat{U})] \geq 1 - \delta.
\label{eq:certified_lower}
\end{eqnarray}
At equal CGP, the generator with smaller CGPD has a larger distribution-independent performance certificate. This suggests a risk-sensitive design functional
\begin{eqnarray}
\mathcal{J}_\eta(\hat{U}) = \CB(\hat{U}) - \eta\DCB(\hat{U}) \quad (\eta \geq 0),
\label{eq:mean_deviation_objective}
\end{eqnarray}
where $\eta$ controls the tradeoff between average resource production and input robustness. By contrast, maximizing CGPD or $\mathcal S_B$ at fixed CGP is appropriate when selective activation is the desired behavior. When higher cumulants are available, sharper tail estimates can be constructed from $K_{\hat{U}}(\lambda)$.

\subsection{Finite-sample CGP benchmarking}

Assume $N$ independent incoherent inputs are sampled and the generated coherence is measured for each one. The empirical estimator is
\begin{eqnarray}
\widetilde{C}_N&=&\frac1N\sum_{n=1}^NX_{B,\hat{U}}(\pp_n).
\label{eq:empirical_cgp}
\end{eqnarray}
It is unbiased and has
\begin{eqnarray}
\Var(\widetilde{C}_N)&=&\frac{\DCB(\hat{U})^2}{N}.
\label{eq:sampling}
\end{eqnarray}
Therefore CGPD is the single-input noise scale governing the standard error of CGP benchmarking. The root-mean-square precision condition $\sqrt{\mathbb E[(\widetilde{C}_N-\CB)^2]}\leq\varepsilon$ requires
\begin{eqnarray}
N&\geq&\frac{\DCB(\hat{U})^2}{\varepsilon^2}.
\label{eq:rms_samples}
\end{eqnarray}
Using the bounded range $0\leq X\leq(d-1)/d$, Bernstein-type confidence bounds can combine the observed variance with the known range. Thus CGPD is both a physical fluctuation and the statistical cost of estimating the mean resource generation.

\subsection{Estimating CGPD and diagnosing non-Gaussian profiles}

The deviation itself can be estimated from the same data. The unbiased sample variance and its square root are
\begin{eqnarray}
\widetilde{V}_N &=& \frac{1}{N-1}\sum_{n=1}^N [X_{B,\hat{U}}(\pp_n) - \widetilde{C}_N]^2, \nonumber \\
\widetilde{\Delta}_N &=& \sqrt{\widetilde{V}_N}.
\label{eq:sample_deviation}
\end{eqnarray}
Because $X_{B,\hat{U}}$ is bounded, all moments exist. If $\mu_4=\mathbb E[(X_{B,\hat{U}}-\CB)^4]$, the standard asymptotic variance formula gives
\begin{eqnarray}
\sqrt N(\widetilde{V}_N - \DCB^2) \Longrightarrow \mathcal N(0,\mu_4-\DCB^4).
\label{eq:sample_variance_clt}
\end{eqnarray}
For $\DCB>0$, the delta method then yields an asymptotic variance $(\mu_4-\DCB^4)/(4\DCB^2)$ for $\widetilde{\Delta}_N$. This exposes a natural hierarchy: CGPD controls estimation of CGP, while the fourth central moment controls estimation of CGPD. Strong disagreement between the empirical histogram and a mean--variance description is itself evidence that skewness or higher cumulants should be retained.

\subsection{Susceptibility to an input-ensemble bias}

Consider the exponentially tilted ensemble
\begin{eqnarray}
\dd\nu_{d,\lambda}(\pp) = \frac{e^{\lambda X_{B,\hat{U}}(\pp)}}{Z_{\hat{U}}(\lambda)}\dd\nu_d(\pp),
\label{eq:tilted_measure}
\end{eqnarray}
which favors coherence-productive inputs for $\lambda>0$ and suppresses them for $\lambda<0$. Define the tilted average
\begin{eqnarray}
C_{B,\lambda}(\hat{U}) = \mathbb E_{\nu_{d,\lambda}}[X_{B,\hat{U}}].
\label{eq:tilted_average}
\end{eqnarray}

\begin{proposition}[Coherence-generation susceptibility]
\label{prop:susceptibility}
The tilted mean satisfies
\begin{eqnarray}
C_{B,0}(\hat{U}) &=& \CB(\hat{U}), \nonumber \\
\left.\frac{\dd C_{B,\lambda}(\hat{U})}{\dd\lambda}\right|_{\lambda=0} &=& \DCB(\hat{U})^2.
\label{eq:susceptibility}
\end{eqnarray}
\end{proposition}

\begin{proof}---By Eq.~(\ref{eq:cumulant_generator}), $C_{B,\lambda}(\hat{U})=K_{\hat{U}}'(\lambda)$. Differentiating once more gives the variance in the tilted ensemble; at $\lambda=0$ this is $\DCB(\hat{U})^2$.
\end{proof}

CGPD squared is therefore a linear-response coefficient. It measures how rapidly the average coherence generation changes when the input distribution is infinitesimally biased toward favorable populations. This interpretation is especially natural in variational state preparation or data-encoding tasks, where the input distribution can be trained or reweighted.

\subsection{Four-copy measurement representation}

Let $\hat{S}$ be the swap operator on two copies and define
\begin{eqnarray}
\hat{K}_B = \hat{S}-\DB^{\otimes2}(\hat{S}).
\label{eq:KB}
\end{eqnarray}
For any state $\hat{\sigma}$,
\begin{eqnarray}
\cB(\hat{\sigma}) = \Tr[\hat{K}_B\hat{\sigma}^{\otimes2}].
\label{eq:swap_coherence}
\end{eqnarray}
The CGP protocol of Ref.~\cite{ZanardiStyliarisVenuti2017} follows from averaging this two-copy expression. The second moment requires four copies.

\begin{proposition}[Four-copy representation]
\label{prop:fourcopy}
The second moment of unitary coherence generation is
\begin{eqnarray}
M_2(\hat{U}) = \Tr\!\left[(\hat{K}_B\otimes \hat{K}_B) \hat{U}^{\otimes4}\hat{\Omega}_B^{(4)}\hat{U}^{\dagger\otimes4}\right].
\label{eq:four_copy}
\end{eqnarray}
\end{proposition}

\begin{proof}---For $\hat{\sigma}_{\pp}=\hat{U}\hat{\rho}_{\pp}\hat{U}^\dagger$, Eq.~(\ref{eq:swap_coherence}) gives $X_{B,\hat{U}}(\pp)=\Tr[\hat{K}_B\hat{\sigma}_{\pp}^{\otimes2}]$. The square of this trace is a four-copy trace with $\hat{K}_B \otimes \hat{K}_B$. Averaging $\hat{\rho}_{\pp}^{\otimes4}$ produces $\hat{\Omega}_B^{(4)}$ and yields Eq.~(\ref{eq:four_copy}).
\end{proof}

The terms in $\hat{K}_B\otimes \hat{K}_B$ are products of swap and dephased-swap observables. Hence, $M_2$ and CGPD can be accessed without reconstructing the full channel or the full coherence distribution. The required input moment state is classically correlated after dephasing; entangled input preparation is not essential, in direct analogy with the original CGP protocol.

A randomized-measurement implementation is also possible: fourth-order functions of outcome probabilities can estimate the same nonlinear invariant using classical postprocessing~\cite{Elben2018,Brydges2019}. The multi-copy formula is retained here because it makes the operator structure transparent.

\section{Extension to unital quantum channels}\label{sec:channels}

\subsection{Gram-matrix extension}

The population-space formulation is not restricted to unitaries. Let $\EE$ be a completely positive trace-preserving (CPTP) and unital map, $\EE(\hat{\openone})=\hat{\openone}$~\cite{NielsenChuang2010}. Define the coherent output associated with the $i$th basis projector,
\begin{eqnarray}
\hat{A}_i^{(\EE)} = \QB\!\left[\EE(\ketbra{i}{i})\right],
\label{eq:channel_vectors}
\end{eqnarray}
and the real Gram matrix
\begin{eqnarray}
(\hat{\Gamma}_{\EE})_{ij} = \Tr\!\left[\hat{A}_i^{(\EE)}\hat{A}_j^{(\EE)}\right].
\label{eq:channel_gram}
\end{eqnarray}

\begin{proposition}[Unital-channel coherence matrix]
\label{prop:channel_gamma}
For every unital channel $\EE$,
\begin{eqnarray}
\cB[\EE(\hat{\rho}_{\pp})]&=&\pp^{\mathsf T}\hat{\Gamma}_{\EE}\pp.
\label{eq:channel_quadratic}
\end{eqnarray}
The matrix $\hat{\Gamma}_{\EE}$ is positive semidefinite and satisfies $\hat{\Gamma}_{\EE} \vone = 0$.
\end{proposition}

\begin{proof}---Linearity gives $\QB\EE(\hat{\rho}_{\pp})=\sum_ip_i\hat{A}_i^{(\EE)}$. Taking the squared Hilbert--Schmidt norm yields Eq.~(\ref{eq:channel_quadratic}), and the Gram form proves positivity. The unitality implies
\begin{eqnarray}
\sum_i\hat{A}_i^{(\EE)} = \QB[\EE(\hat{\openone})] = \QB(\hat{\openone}) = 0,
\end{eqnarray}
so every row sum of $\hat{\Gamma}_{\EE}$ vanishes.
\end{proof}

\begin{theorem}[CGPD for unital channels]
\label{thm:channel_moments}
Eqs.~(\ref{eq:m1}), (\ref{eq:m2}), (\ref{eq:cgpd_closed}), (\ref{eq:lower_bound}), and (\ref{eq:variance_decomposition}) remain valid for a unital channel after replacing $\hat{\Gamma}_{\hat{U}}$ by $\hat{\Gamma}_{\EE}$.
\end{theorem}

\begin{proof}---The derivation of the moments uses only the quadratic form, symmetry of the matrix, and the zero-row-sum property. Proposition~\ref{prop:channel_gamma} supplies all three. The decomposition in Theorem~\ref{thm:decomposition} is purely algebraic, and its nonnegative remainder yields the lower bound in Theorem~\ref{thm:lower_bound}.
\end{proof}

\subsection{Faithfulness and the role of unitality}

For channels, vanishing CGP and CGPD characterize the absence of coherence generation from diagonal inputs,
\begin{eqnarray}
\CB(\EE)=\DCB(\EE)=0 \Longleftrightarrow \QB\EE\DB=0.
\label{eq:channel_faithful}
\end{eqnarray}
This condition is weaker than $[\EE,\DB]=0$. A channel may erase input coherence in a way that fails to commute with dephasing while still mapping every incoherent state to an incoherent state. Thus, the deviation remains faithful to the operational task of coherence generation, but not to every stronger definition of an incoherent channel. This distinction agrees with the corresponding caveat for average CGP~\cite{ZanardiStyliarisVenuti2017,ZanardiStyliarisVenuti2017b}.

Unitality is also the reason for the universal lower bound. For a general trace-preserving channel define $\hat{A}_i=\QB\EE(\ketbra{i}{i})$ and
\begin{eqnarray}
\hat{a}_0 = \frac1d\sum_i\hat{A}_i=\QB\EE(\hat{\openone}/d).
\label{eq:nonunital_offset}
\end{eqnarray}
Writing $\pp=\vone/d+\xx$ gives
\begin{eqnarray}
\QB\EE(\hat{\rho}_{\pp}) &=& \hat{a}_0+\sum_ix_i\hat{A}_i, \nonumber \\
\cB[\EE(\hat{\rho}_{\pp})] &=& \|\hat{a}_0\|_2^2+2\sum_ix_i\operatorname{Re}\Tr(\hat{a}_0^\dagger \hat{A}_i) \nonumber \\
	&& \quad +\sum_{ij}x_ix_j\Tr(\hat{A}_i^\dagger \hat{A}_j).
\label{eq:nonunital_affine}
\end{eqnarray}
A coherent offset at the simplex center and a linear term are therefore present. The moment problem remains finite, but it is no longer governed by a zero-row-sum quadratic form, and Eq.~(\ref{eq:lower_bound}) need not hold. This identifies a mathematically clean boundary of the present theory and a natural route to a separate nonunital CGPD framework.

\subsection{Kraus representation}

If $\EE(\hat{\rho})=\sum_k\hat{E}_k\hat{\rho}\hat{E}_k^\dagger$, Eq.~(\ref{eq:channel_gram}) becomes
\begin{eqnarray}
(\hat{\Gamma}_{\EE})_{ij}&=&\sum_{l\neq m}
\left[\sum_k(\hat{E}_k)_{li}(\hat{E}_k)_{mi}^*\right]^*
\left[\sum_r(\hat{E}_r)_{lj}(\hat{E}_r)_{mj}^*\right].
\label{eq:kraus_gram}
\end{eqnarray}
This makes clear that CGPD probes correlations between the coherent outputs of distinct basis populations, not only their individual norms.

\section{Examples in quantum simulation and feature maps}\label{sec:examples}

In this section, we develop a hierarchy of matched-CGP examples. A four-level calibration isolates the elementary mechanism, an analytically solvable quantum walk resolves graph topology, and a quasiperiodic transport simulation contrasts extended with localized dynamics. In every case, the mean coherence alone hides a physically meaningful difference in input selectivity.

\subsection{Matching protocol and four-level calibration}

An equal-CGP comparison fixes a target $C_*$, solves $\CB(\hat{U}_1)=\CB(\hat{U}_2)=C_*$, and then compares $\DCB$, $\mathcal{S}_B$, and, when useful, the full push-forward distributions. This protocol removes average strength as a confounding variable. The following gate-level pair serves as a calibration before the Hamiltonian examples. Let
\begin{eqnarray}
\hat{R}(\theta) = \begin{pmatrix}
\cos\theta  &  -\sin\theta \\
\sin\theta  &  \cos\theta
\end{pmatrix},
\label{eq:rotation}
\end{eqnarray}
and define
\begin{eqnarray}
\hat{U}_{\rm loc} &=& \hat{R}(\pi/4) \oplus \hat{\openone}_2,\nonumber\\
\hat{U}_{\rm dis} &=& \hat{R}(\pi/8) \oplus \hat{R}(\pi/8).
\label{eq:block_example}
\end{eqnarray}

\begin{proposition}[Localized versus distributed population modes]
The two four-level unitaries in Eq.~(\ref{eq:block_example}) satisfy
\begin{eqnarray}
\CB(\hat{U}_{\rm loc})&=&\CB(\hat{U}_{\rm dis})=\frac1{20},
\label{eq:block_equal_mean}
\end{eqnarray}
but
\begin{eqnarray}
\DCB(\hat{U}_{\rm loc})^2&=&\frac{13}{2800},\nonumber\\
\DCB(\hat{U}_{\rm dis})^2&=&\frac1{600}.
\label{eq:block_diff_var}
\end{eqnarray}
Their landscapes are
\begin{eqnarray}
X_{\rm loc}(\pp)&=&\frac12(p_1-p_2)^2,\nonumber\\
X_{\rm dis}(\pp)&=&\frac14\left[(p_1-p_2)^2+(p_3-p_4)^2\right].
\label{eq:block_landscapes}
\end{eqnarray}
\end{proposition}

\begin{proof}---For a two-level rotation, the response eigenvalue on the population-contrast mode is $\sin^2(2\theta)$. Hence,
\begin{eqnarray}
\hat{\Gamma}_{\rm loc} &=& \frac{1}{2}
\begin{pmatrix}
1 & -1 & 0 & 0 \\
-1 & 1 & 0 & 0 \\
0 & 0 & 0 & 0 \\
0 & 0 & 0 &0
\end{pmatrix}, \nonumber \\
\hat{\Gamma}_{\rm dis} &=& \frac{1}{4}
\begin{pmatrix}
1 & -1 & 0 & 0 \\
-1 & 1 & 0 & 0 \\
0 & 0 & 1 & -1 \\
0 & 0 & -1 & 1
\end{pmatrix}
\label{eq:block_gamma}
\end{eqnarray}
Both traces equal one, which gives Eq.~(\ref{eq:block_equal_mean}). The first matrix has $\Tr\hat{\Gamma}^2=1$ and squared diagonal norm $1/2$; the second has $\Tr\hat{\Gamma}^2=1/2$ and squared diagonal norm $1/4$. Theorem~\ref{thm:moments} gives Eq.~(\ref{eq:block_diff_var}), and direct contraction gives Eq.~(\ref{eq:block_landscapes}).
\end{proof}

The pair makes the interpretation unambiguous. $\hat{U}_{\rm loc}$ responds strongly to one contrast and ignores the remaining two population directions. $\hat{U}_{\rm dis}$ spreads the same trace response over two contrasts and is therefore less variable.

\subsection{Continuous-time quantum walks}

For a graph $G$ with adjacency operator $\hat{A}_G$, consider the spectrally normalized Hamiltonian and propagator
\begin{eqnarray}
\hat{H}_G&=&\frac{\hat{A}_G}{\|\hat{A}_G\|_2},\qquad
\hat{U}_G(t)=e^{-it\hat{H}_G}.
\label{eq:ctqw}
\end{eqnarray}
The continuous-time quantum walks are standard models for transport, algorithms, and Hamiltonian simulation~\cite{FarhiGutmann1998,Childs2009,VenegasAndraca2012,ChenLiLi2024}. In the vertex basis, an incoherent state is a classical distribution over vertices. CGP measures average conversion of vertex uncertainty into coherent amplitudes, whereas CGPD measures how strongly this conversion depends on the initial vertex population profile.

\begin{figure}[t]
\centering
\includegraphics[width=0.46\textwidth]{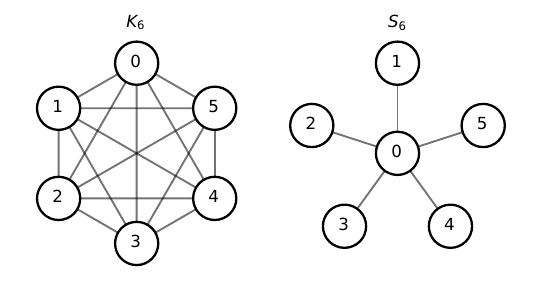}
\caption{The graph structures used in the continuous-time quantum-walk comparison. The complete graph $K_6$ has all pairwise vertex connections, whereas the star graph $S_6$ consists of one hub vertex connected to five leaves. The vertex labels match those used in Eqs.~(\ref{eq:landscape_complete}) and~(\ref{eq:landscape_star}).}
\label{fig:walk_graphs}
\end{figure}

Here, we compare the complete graph $K_6$ and the star graph $S_6$ with one hub and five leaves (see Fig.~\ref{fig:walk_graphs}). Their times are chosen to produce exactly the same CGP.

\begin{proposition}[Topology-resolved equal-CGP pair]
\label{prop:walk}
Let
\begin{eqnarray}
t_K &=& \frac56\arccos\!\left(\frac{3\sqrt{145}}{25}-2\right), \nonumber\\
t_S &=& \pi.
\label{eq:walk_times}
\end{eqnarray}
Then,
\begin{eqnarray}
\CB[\hat{U}_{K_6}(t_K)] = \CB[\hat{U}_{S_6}(t_S)] = \frac{16}{175},
\label{eq:walk_same_mean}
\end{eqnarray}
but
\begin{eqnarray}
\DCB[\hat{U}_{K_6}(t_K)] &=& \frac{16\sqrt{10}}{875} \simeq 0.0578245, \nonumber\\
\DCB[\hat{U}_{S_6}(t_S)] &=& \frac{4\sqrt{230}}{875} \simeq0.0693291.
\label{eq:walk_different_deviation}
\end{eqnarray}
The star-graph variance is $23/16$ times the complete-graph variance.
\end{proposition}

\begin{proof}---For $K_6$, let $\hat{\Pi}_6=\vone\vone^{\mathsf T}/6$. Spectral decomposition of $\hat{H}_{K_6}$ gives
\begin{eqnarray}
\hat{P}_{K_6}(t) &=& \hat{\Pi}_6+r_K(t)(\one-\hat{\Pi}_6), \nonumber \\
r_K(t) &=& \frac{2+\cos(6t/5)}{3}.
\end{eqnarray}
At $t=t_K$, $r_K(t_K)^2=29/125$, and hence
\begin{eqnarray}
\hat{\Gamma}_{K_6}(t_K) = \frac{96}{125}(\one-\hat{\Pi}_6).
\label{eq:gamma_complete}
\end{eqnarray}
For $S_6$, write $\ket{s}=5^{-1/2}\sum_{j=1}^5\ket{j}$ for the uniform leaf state. The normalized Hamiltonian couples only $\ket{0}$ and $\ket{s}$. At $t=\pi$, the hub and uniform leaf mode acquire a sign while the orthogonal leaf modes are stationary, yielding
\begin{eqnarray}
\hat{\Gamma}_{S_6}(\pi) = \frac{24}{25} \left[0\oplus\left(\one_5-\frac{\vone_5\vone_5^{\mathsf T}}5\right)\right].
\label{eq:gamma_star}
\end{eqnarray}
Both matrices have trace $96/25$, proving Eq.~(\ref{eq:walk_same_mean}). Substitution into Eq.~(\ref{eq:cgpd_closed}) gives Eq.~(\ref{eq:walk_different_deviation}). Further spectral details are given in Appendix~\ref{app:walk}.
\end{proof}

The landscapes are
\begin{eqnarray}
X_{K_6}(\pp) = \frac{96}{125} \left(\sum_{j=0}^5p_j^2 - \frac{1}{6}\right),
\label{eq:landscape_complete}
\end{eqnarray}
and
\begin{eqnarray}
X_{S_6}(\pp) = \frac{24}{25} \left[\sum_{j=1}^5p_j^2 - \frac{1}{5}\left(\sum_{j=1}^5p_j\right)^2\right].
\label{eq:landscape_star}
\end{eqnarray}
The complete graph treats all zero-sum population directions identically and saturates the universal bound. The star graph is blind to the hub population and to the uniform leaf mode, while strongly converting contrasts among leaves. Equal trace gives equal CGP, but the missing modes create anisotropy and a larger CGPD.

The corresponding time traces and generated-coherence distributions are shown in Fig.~\ref{fig:walk}. The marked propagators have the same mean, while the visibly broader $S_6$ distribution resolves the topology-induced selectivity.

\begin{figure}[t]
\centering
\includegraphics[width=0.46\textwidth]{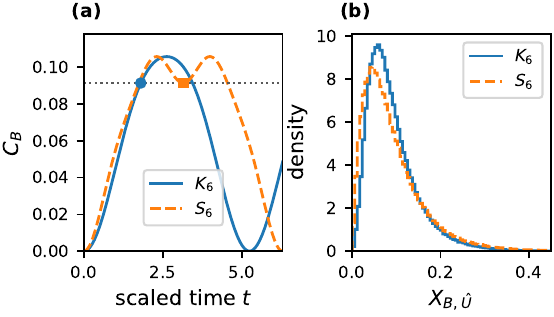}
\caption{The continuous-time quantum walks on $K_6$ and $S_6$. (a) CGP versus scaled time; marked propagators have the common value $16/175$. (b) Generated-coherence distributions for uniform-simplex inputs at the marked times. The means coincide, while the star graph has the broader distribution because coherence generation is concentrated on leaf-population contrasts. Histograms use $3.5\times10^5$ Dirichlet samples; exact moments are given in Proposition~\ref{prop:walk}.}
\label{fig:walk}
\end{figure}

The example has a direct graph-learning interpretation. Quantum evolution kernels and graph feature maps encode topology through Hamiltonian time evolution~\cite{Rossi2015,Henry2021,Albrecht2023}. If incoherent populations represent uncertain or data-dependent vertex weights, CGP measures average coherence injected into the feature state, while CGPD measures sensitivity to those weights. Two feature maps with equal mean coherence can therefore have different robustness and different emphasis on graph modes.

\subsection{Clean and quasiperiodic transport}

We next consider a single particle on an open chain of $L=6$ sites. The clean and Aubry--Andr\'e--Harper (AAH) Hamiltonians are
\begin{eqnarray}
\hat{H}_{\rm cl} = -J\sum_{n=0}^{L-2}\left(\ketbra{n}{n+1}+\ketbra{n+1}{n}\right),
\label{eq:clean_hamiltonian}
\end{eqnarray}
\begin{eqnarray}
\hat{H}_{\rm AAH} = \hat{H}_{\rm cl}+\lambda\sum_{n=0}^{L-1}\cos(2\pi\alpha n+\phi)\ketbra{n}{n},
\label{eq:aah_hamiltonian}
\end{eqnarray}
with $\alpha=(\sqrt5-1)/2$, $\phi=0$, and $\lambda/J=5$. The infinite AAH model is in the localized regime for $\lambda/J>2$; finite chains retain the corresponding suppression and mode selectivity of transport~\cite{AubryAndre1980,BiddleDasSarma2009,Lahini2009}.

For $\hat{U}_\mu(t)=e^{-it\hat{H}_\mu}$, we choose one time on each trajectory by solving
\begin{eqnarray}
\CB[\hat{U}_\mu(t_\mu)] = 0.035.
\label{eq:aah_matching}
\end{eqnarray}
On the selected branches,
\begin{eqnarray}
Jt_{\rm cl} &=& 4.1254706980, \nonumber \\
Jt_{\rm AAH} &=& 6.7980261219.
\label{eq:aah_times}
\end{eqnarray}
The matched deviations are
\begin{eqnarray}
\DCB[\hat{U}_{\rm cl}(t_{\rm cl})] &=& 0.0240471626, \nonumber \\
\DCB[\hat{U}_{\rm AAH}(t_{\rm AAH})] &=& 0.0374482254.
\label{eq:aah_deviations}
\end{eqnarray}
Thus, the localized model has the same average generated coherence but a variance approximately $2.43$ times larger. Appendix~\ref{app:numerics} gives the numerical protocol, branch checks, and reproducibility details for this comparison.

\begin{figure}[t]
\centering
\includegraphics[width=0.46\textwidth]{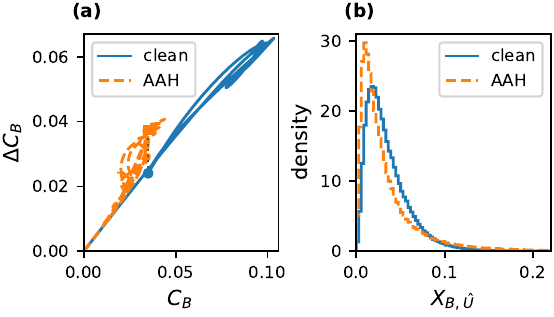}
\caption{Matched-CGP transport example for a clean chain and an AAH chain with $L=6$ and $\lambda/J=5$. (a) The parametric trajectories in the CGP--CGPD plane for $0\leq Jt\leq8$. The markers identify times with the common $\text{CGP}=0.035$; the dotted segment is the difference resolved by CGPD. (b) The distributions of generated coherence at the matched times from $3.5\times10^5$ uniform-simplex samples. The AAH distribution is more concentrated near weak generation but has a longer tail, yielding a substantially larger deviation at the same mean.}
\label{fig:aah}
\end{figure}

Fig.~\ref{fig:aah} explains the mechanism. Extended clean dynamics spreads its coherence-conversion capability over several population modes, whereas quasiperiodic localization produces a more heterogeneous response. Many input profiles generate little coherence, while a smaller set aligned with resonant local contrasts produces much larger values. The long tail compensates the enhanced weight near zero so that the average can match the clean chain. CGPD exposes this redistribution immediately.

The example also illustrates why CGPD can be useful in Hamiltonian feature maps. In kernel-based quantum learning, a data point may determine a Hamiltonian parameter, a graph, an input population profile, or a fixed-basis readout distribution~\cite{Havlicek2019,SchuldKilloran2019,Baker2023}. Matching average coherence does not guarantee comparable encoding stability. A feature layer with larger CGPD is more selective to population uncertainty and may emphasize particular modes. This statement concerns coherence-sensitive readout or data-dependent dynamics; a single data-independent unitary followed by a unitary-invariant full-state kernel would, of course, leave pairwise state overlaps unchanged. With that caveat, $(\CB,\DCB)$ separates average resource generation from input-dependent robustness in a way that conventional kernel expressivity measures do not.

\section{Discussion}\label{sec:discussion}

We have introduced the coherence-generating power deviation (CGPD) as the standard deviation of Hilbert--Schmidt coherence generated over the uniform simplex of incoherent inputs. The population-space matrix $\hat{\Gamma}_{\hat{U}}=\one-\hat{P}_{\hat{U}}^{\mathsf T}\hat{P}_{\hat{U}}$ turns unitary coherence generation into a quadratic landscape over population imbalances, whose mean is CGP and whose width is CGPD. This formulation gave exact first- and second-moment formulas in arbitrary finite dimension, identified the additional invariants beyond the trace that enter the deviation, proved a universal CGP--CGPD lower bound, and decomposed the deviation into an unavoidable isotropic part and an excess selective part. We also showed that the same moment geometry extends to unital channels and illustrated the theory through matched-CGP examples in graph quantum walks and quasiperiodic transport.

The main implication is that the average coherence generation alone is not sufficient to characterize a dynamical resource. Two processes can have identical CGP while distributing their coherence-generation ability over population space in very different ways. CGPD separates nearly isotropic generation from selective generation concentrated on particular population modes. This distinction is important whenever the diagonal input is uncertain, sampled, trained, or data dependent: a small excess selectivity factor indicates a stable coherence source, whereas a larger one identifies dynamics that preferentially amplifies specific population imbalances.

CGPD is also operationally useful. It is the single-input variance scale for finite-sample CGP benchmarking, gives distribution-independent reliability bounds for unknown incoherent inputs, and equals the linear susceptibility of the mean coherence generation to an input-ensemble bias. The four-copy formula shows that this second moment can be accessed without full process tomography. In the examples, graph topology and localization changed CGPD at fixed CGP, demonstrating that the deviation can reveal the structural information invisible to the mean and can serve as a resource-sensitive diagnostic for quantum simulation and Hamiltonian feature maps.


\begin{acknowledgments}
This work was supported by the National Research Foundation of Korea through grants funded by the Ministry of Science and ICT (RS-2024-00432214, RS-2025-03532992, and RS-2025-18362970); by the Institute of Information and Communications Technology Planning and Evaluation through a grant funded by the Korean government (RS-2019-II190003, ``Research and Development of Core Technologies for Programming, Running, Implementing and Validating of Fault-Tolerant Quantum Computing System''); by the Korean ARPA-H Project through the Korea Health Industry Development Institute (KHIDI), funded by the Ministry of Health \& Welfare, Republic of Korea (RS-2025-25456722); and by Grant No.~K25L5M2C2 at the Korea Institute of Science and Technology Information. The authors acknowledge the Yonsei University Quantum Computing Project Group for support and access to the IBM Quantum System One operated at Yonsei University.
\end{acknowledgments}

\appendix

\section{Uniform-simplex moments from Haar integration}
\label{app:haar}

This appendix derives the moment identities used in Sec.~\ref{sec:moments}. For a Haar-random vector $\ket{\psi}\in\mathbb C^d$,
\begin{eqnarray}
\mathbb E_{\psi}[\ketbra{\psi}{\psi}^{\otimes n}] &=& \frac{\hat{\Pi}_{\mathrm{sym}}^{(n)}}{\dim\operatorname{Sym}^n(\mathbb C^d)} \nonumber\\
	&=& \frac{1}{(d)_n}\sum_{\pi\in S_n}\hat{T}_\pi.
\label{eq:haar_state_moment}
\end{eqnarray}
The second equality follows from
\begin{eqnarray}
&& \hat{\Pi}_{\mathrm{sym}}^{(n)} = \frac1{n!}\sum_{\pi\in S_n}\hat{T}_\pi, \nonumber \\
&& \dim\operatorname{Sym}^n(\mathbb C^d) = \frac{(d)_n}{n!}.
\end{eqnarray}
Applying $\DB^{\otimes n}$ gives Eq.~(\ref{eq:tensor_moment}).

\begin{lemma}[Dirichlet monomial moments]
\label{lem:dirichlet}
For nonnegative integers $m_1,\ldots,m_d$ with $m=\sum_im_i$,
\begin{eqnarray}
\mathbb E_{\nu_d}\!\left[\prod_{i=1}^dp_i^{m_i}\right] = \frac{(d-1)!\prod_i m_i!}{(d+m-1)!}.
\label{eq:dirichlet_monomial}
\end{eqnarray}
\end{lemma}

\begin{proof}---The uniform density on $\Delta_{d-1}$ is $(d-1)!$. The multivariate beta integral gives
\begin{eqnarray}
\int_{\Delta_{d-1}}\prod_i p_i^{m_i}\dd\pp = \frac{\prod_i\Gamma(m_i+1)}{\Gamma(d+m)}.
\end{eqnarray}
Multiplication by $(d-1)! = \Gamma(d)$ yields Eq.~(\ref{eq:dirichlet_monomial}).
\end{proof}

Eq.~(\ref{eq:dirichlet_second}) follows by setting $m=2$. For fourth order, the possible index multiplicities are
\begin{eqnarray}
\mathbb E[p_i^4]&=&\frac{24}{\Dfour}, \nonumber \\
\mathbb E[p_i^3p_j]&=&\frac6{\Dfour}\quad(i\neq j), \nonumber \\
\mathbb E[p_i^2p_j^2]&=&\frac4{\Dfour}\quad(i\neq j), \nonumber \\
\mathbb E[p_i^2p_jp_k]&=&\frac2{\Dfour}\quad(i,j,k\ \mathrm{distinct}), \nonumber \\
\mathbb E[p_ip_jp_kp_l]&=&\frac1{\Dfour}\quad(i,j,k,l\ \mathrm{distinct}).
\label{eq:fourth_patterns}
\end{eqnarray}
These values are equivalent to the permutation in Eq.~(\ref{eq:dirichlet_general}).

\section{Fourth-moment contraction}\label{app:fourth}

Let $\hat{G}$ be any real symmetric matrix satisfying $\hat{G}\vone=0$. We evaluate
\begin{eqnarray}
\mathbb E[(\pp^{\mathsf T}\hat{G}\pp)^2] = \sum_{ijkl}\hat{G}_{ij}\hat{G}_{kl}\mathbb E[p_ip_jp_kp_l].
\label{eq:generic_second}
\end{eqnarray}
Using Eq.~(\ref{eq:dirichlet_general}), each permutation in $S_4$ imposes index identifications. The conjugacy classes and their contributions are summarized in Table~\ref{tab:s4}.

\begin{table}[t]
\caption{Conjugacy classes of $S_4$ in the fourth-moment contraction.}
\label{tab:s4}
\begin{ruledtabular}
\begin{tabular}{lccc}
Cycle type & Number & Fixed point? & Total contribution \\
\hline
$1^4$ & $1$ & yes & $0$ \\
$2\,1^2$ & $6$ & yes & $0$ \\
$3\,1$ & $8$ & yes & $0$ \\
$2^2$ & $3$ & no & $(\Tr\hat{G})^2+2\Tr(\hat{G}^2)$ \\
$4$ & $6$ & no & $6\sum_i\hat{G}_{ii}^2$ \\
\end{tabular}
\end{ruledtabular}
\end{table}

Any fixed point leaves one index appearing in a single matrix factor and therefore produces a row sum $\sum_j\hat{G}_{ij}=0$. The three double transpositions are
\begin{eqnarray}
(12)(34),\qquad(13)(24),\qquad(14)(23).
\end{eqnarray}
Their contractions are, respectively,
\begin{eqnarray}
\sum_{ik}\hat{G}_{ii}\hat{G}_{kk} &=& (\Tr\hat{G})^2, \nonumber \\
\sum_{ij}\hat{G}_{ij}^2 &=& \Tr(\hat{G}^2), \nonumber \\
\sum_{ij}\hat{G}_{ij}\hat{G}_{ji} &=& \Tr(\hat{G}^2).
\end{eqnarray}
Each four-cycle identifies all four external indices after contraction and gives $\sum_i\hat{G}_{ii}^2$. There are six such cycles. Dividing by $\Dfour$ proves Eq.~(\ref{eq:m2}).

An alternative derivation follows directly from Eq.~(\ref{eq:fourth_patterns}) by grouping terms according to index multiplicity. The permutation method is shorter and generalizes naturally to higher cumulants.

\section{Details of the fluctuation bound and decomposition}\label{app:bound}

Let $\hat{G} \geq 0$, $\hat{G}\vone=0$, and $t=\Tr\hat{G}$. Define
\begin{eqnarray}
\bar\lambda=\frac{t}{d-1},
\quad
\hat{A}=\hat{G}-\bar\lambda\Hperp.
\label{eq:app_decomp}
\end{eqnarray}
Then, $\hat{A}\vone=0$ and $\Tr\hat{A}=0$. Since $\Hperp^2=\Hperp$ and $\Tr\Hperp=d-1$,
\begin{eqnarray}
\Tr(\hat{G}^2) &=& \bar\lambda^2(d-1)+\Tr(\hat{A}^2) \nonumber \\
	&=& \frac{t^2}{d-1}+\|\hat{A}\|_F^2.
\label{eq:app_spectral}
\end{eqnarray}
The diagonal of $\bar\lambda\Hperp$ is constant and equal to $t/d$, so
\begin{eqnarray}
\sum_j\hat{G}_{jj}^2 = \frac{t^2}{d}+\sum_j\hat{A}_{jj}^2.
\label{eq:app_diagonal}
\end{eqnarray}
Substituting Eqs.~(\ref{eq:app_spectral}) and (\ref{eq:app_diagonal}) into the generic variance formula yields
\begin{eqnarray}
\Var(\pp^{\mathsf T}\hat{G}\pp)&=&
\frac{4t^2}{(d-1)(d+1)^2(d+2)(d+3)}\nonumber\\
&&+\frac{2\|\hat{A}\|_F^2+6\|\operatorname{diag}\hat{A}\|_2^2}{\Dfour}.
\label{eq:app_variance_identity}
\end{eqnarray}
Because $\mathbb E[\pp^{\mathsf T}\hat{G}\pp]=t/[d(d+1)]$, the first term is $\gamma_d^2C^2$. This proves Theorem~\ref{thm:decomposition} and immediately implies the lower bound.

The equality condition is especially transparent in this form: the lower bound is saturated exactly when $\hat{A}=0$, namely when all population modes in $\vone^\perp$ have the same response strength. For a unitary, this condition is weaker than maximal CGP; any feasible $\hat{\Gamma}_{\hat{U}}=\lambda\Hperp$ with $0<\lambda\leq1$ is isotropic, while maximal CGP corresponds to $\lambda=1$.

\section{Analytical details for the graph walks}\label{app:walk}

For the complete graph $K_d$, the adjacency operator is $\hat{A}=\hat{J}_d-\hat{\openone}$, where $\hat{J}_d=\sum_{i,j=1}^d\ketbra{i}{j}$. With spectral normalization $\hat{H}=\hat{A}/(d-1)$,
\begin{eqnarray}
\hat{U}_{K_d}(t) = e^{-it}\hat{\Pi}_d+e^{i \frac{t}{(d-1)}}(\hat{\openone}-\hat{\Pi}_d),
\quad
\hat{\Pi}_d = \frac{\hat{J}_d}{d}.
\label{eq:complete_propagator}
\end{eqnarray}
The transition probabilities are expressed with the real population-space projector $\hat{\Pi}_d=\vone\vone^{\mathsf T}/d$ and have one diagonal value and one off-diagonal value. Direct multiplication gives
\begin{eqnarray}
\hat{P}_{K_d}(t) &=& \hat{\Pi}_d+r_d(t)(\one-\hat{\Pi}_d), \nonumber \\
r_d(t) &=& \frac{d-2+2\cos\!\left[dt/(d-1)\right]}{d}.
\label{eq:complete_transition_general}
\end{eqnarray}
For $d=6$, this reduces to the expression used in Proposition~\ref{prop:walk}.

For the star $S_d$, let $\ket{0}$ be the hub and
\begin{eqnarray}
\ket{s} = \frac1{\sqrt{d-1}}\sum_{j=1}^{d-1}\ket{j}.
\end{eqnarray}
The normalized Hamiltonian is
\begin{eqnarray}
\hat{H}_{S_d} = \ketbra{0}{s}+\ketbra{s}{0}.
\label{eq:star_hamiltonian}
\end{eqnarray}
It acts as a Pauli $X$ on $\operatorname{span}\{\ket{0},\ket{s}\}$ and vanishes on the $(d-2)$-dimensional leaf-contrast subspace. At $t=\pi$, the active two-dimensional subspace acquires a sign while the contrast subspace is unchanged. Evaluating the absolute squares gives
\begin{eqnarray}
\hat{\Gamma}_{S_d}(\pi) = \frac{8(d-3)}{(d-1)^2} \left[0\oplus\left(\one_{d-1}-\frac{\hat{J}_{d-1}}{d-1}\right)\right].
\label{eq:star_gamma_general}
\end{eqnarray}
Setting $d=6$ gives Eq.~(\ref{eq:gamma_star}).

\section{Numerical protocol, branch checks, and reproducibility}
\label{app:numerics}

All numerical examples use the exact matrix formulas in Theorem~\ref{thm:moments}; Monte Carlo sampling is used only to visualize the distribution. For a propagator $\hat{U}(t)$, set
\begin{eqnarray}
\hat{P}(t) &=& \abs{\hat{U}(t)}^{\circ2}, \nonumber \\
\hat{\Gamma}(t) &=& \one-\hat{P}(t)^{\mathsf T}\hat{P}(t), \nonumber \\
g(t) &=& \Tr\hat{\Gamma}(t), \nonumber \\
q(t) &=& \Tr[\hat{\Gamma}(t)^2], \nonumber \\
h(t) &=& \sum_j\hat{\Gamma}_{jj}(t)^2.
\label{eq:numerical_invariants}
\end{eqnarray}
Here, $\abs{\cdot}^{\circ2}$ denotes entrywise absolute square. Then,
\begin{eqnarray}
\CB[\hat{U}(t)] &=& \frac{g(t)}{d(d+1)}, \nonumber \\
\DCB[\hat{U}(t)]^2 &=& \frac{g(t)^2+2q(t)+6h(t)}{\Dfour} - \CB[\hat{U}(t)]^2.
\label{eq:numerical_recipe}
\end{eqnarray}

\begin{figure}[t]
\centering
\includegraphics[width=0.46\textwidth]{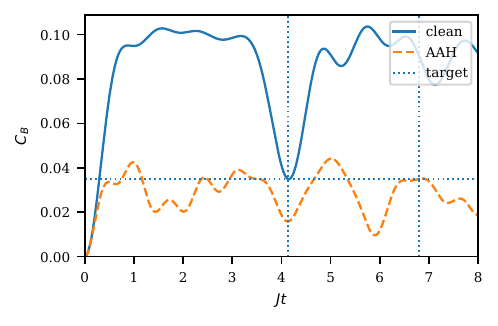}
\caption{The CGP time traces for the clean and AAH chains used in Fig.~\ref{fig:aah}. The horizontal line marks the target CGP $0.035$, and the vertical markers identify the selected roots. Exact matrix moments, rather than histogram estimates, determine all curves.}
\label{fig:aah_traces}
\end{figure}

For the AAH calculation, we set $J=1$, $L=6$, $\alpha=(\sqrt5-1)/2$, $\phi=0$, and $\lambda=5$. The main-text roots are found at high working precision on brackets containing the values in Eq.~(\ref{eq:aah_times}). The time traces in Fig.~\ref{fig:aah_traces} show that the common target is reached on several branches; the quoted pair is a deliberately nontrivial revival-time comparison, not a small-time perturbative coincidence.

To verify that the ordering is not tied to one selected revival, we also apply a deterministic first-crossing rule
\begin{eqnarray}
\tau_\mu(C_*) = \inf\{t>0:\CB[\hat{U}_\mu(t)]=C_*\}.
\label{eq:first_crossing}
\end{eqnarray}
Table~\ref{tab:first_crossing} reports the several targets. The AAH deviation exceeds the clean-chain deviation throughout this range, although the magnitude of the enhancement varies because the finite system has multiple interfering frequencies.

\begin{table}[t]
\caption{First-crossing robustness check for the transport example. The last column is the ratio of variances.}
\label{tab:first_crossing}
\begin{ruledtabular}
\begin{tabular}{cccc}
$C_*$ & $\Delta_{\rm cl}$ & $\Delta_{\rm AAH}$ & $V_{\rm AAH}/V_{\rm cl}$ \\
\hline
$0.025$ & $0.0181115$ & $0.0192111$ & $1.12511$ \\
$0.030$ & $0.0216623$ & $0.0240290$ & $1.23044$ \\
$0.035$ & $0.0251848$ & $0.0389494$ & $2.39180$ \\
$0.040$ & $0.0286763$ & $0.0394245$ & $1.89011$ \\
\end{tabular}
\end{ruledtabular}
\end{table}

The uniform-simplex samples are generated as
\begin{eqnarray}
p_i&=&\frac{z_i}{\sum_jz_j}, \quad z_i\sim\operatorname{Exp}(1),
\label{eq:dirichlet_sampling}
\end{eqnarray}
which is equivalent to $\operatorname{Dir}(1,\ldots,1)$. The histograms use $3.5\times10^5$ common random population vectors for the matched dynamics, reducing visual differences caused by independent sampling. A fixed seed is specified in the code.

\section{Low-dimensional checks}\label{app:lowdim}

For completeness, we collect two calculations that clarify where genuinely new deviation information begins. For a qubit, write $\pp=(x,1-x)$ with $x$ uniform on $[0,1]$. Every unistochastic matrix has one nontrivial singular value, so $\hat{\Gamma}_{\hat{U}} = \lambda\hat{H}_2$ and
\begin{eqnarray}
X_{B,\hat{U}}(x) &=& \frac{\lambda}{2}(2x-1)^2, \nonumber \\
\CB(\hat{U}) &=& \frac{\lambda}{6}, \nonumber \\
X_{\max} &=& \frac{\lambda}{2}=3\CB(\hat{U}).
\label{eq:qubit_distribution_derivation}
\end{eqnarray}
The transformation $y=(2x-1)^2$ has density $1/(2\sqrt y)$ on $(0,1)$, which proves Eq.~(\ref{eq:qubit_density}) and gives $\DCB=2\CB/\sqrt5$. Thus the entire qubit family is isotropic on the one-dimensional space $\vone^\perp$.

For the four-level calibration pair, Eq.~(\ref{eq:block_gamma}) gives
\begin{eqnarray}
(\Tr\hat{\Gamma},\Tr\hat{\Gamma}^2,\|\operatorname{diag}\hat{\Gamma}\|_2^2)_{\rm loc}
&=&(1,1,1/2),\nonumber\\
(\Tr\hat{\Gamma},\Tr\hat{\Gamma}^2,\|\operatorname{diag}\hat{\Gamma}\|_2^2)_{\rm dis}
&=&(1,1/2,1/4).
\label{eq:block_invariants_appendix}
\end{eqnarray}
With $\Dfour=4\cdot5\cdot6\cdot7=840$, the second moments are $1/140$ and $1/240$, respectively. Subtracting the common mean square $1/400$ yields the variances in Eq.~(\ref{eq:block_diff_var}). This calculation is the smallest transparent instance in which equal trace does not fix the remaining second-order invariants.

\section{Relation to entangling power deviation and coherence measures}\label{app:relation}

Entangling power deviation~\cite{ChoBang2026} quantifies the fluctuations of entanglement generated across a specified bipartition from Haar-distributed pure product inputs. CGPD instead quantifies the fluctuations of basis coherence generated from diagonal states over a single probability simplex and requires no bipartition. Both studies share the perspective that a resource-generation profile contains information absent from its mean, and use the related Haar-integration and permutation-contraction techniques. In the present setting, the simplex geometry yields a population-response matrix and an isotropic--selective fluctuation decomposition, with applications to topology and localization in coherence generation.

CGPD is a statistic of a process's coherence-generation profile, rather than an alternative coherence measure for individual states. Its definition here uses squared Hilbert--Schmidt coherence, consistently with probabilistic CGP~\cite{ZanardiStyliarisVenuti2017}. Analogous deviations can be defined using, for example, the $\ell_1$ norm or relative entropy of coherence~\cite{Baumgratz2014,Streltsov2017}; these generally give different profiles, and the present quadratic-response formulas and bounds do not automatically carry over.


\begin{thebibliography}{31}%
\makeatletter
\providecommand \@ifxundefined [1]{%
 \@ifx{#1\undefined}
}%
\providecommand \@ifnum [1]{%
 \ifnum #1\expandafter \@firstoftwo
 \else \expandafter \@secondoftwo
 \fi
}%
\providecommand \@ifx [1]{%
 \ifx #1\expandafter \@firstoftwo
 \else \expandafter \@secondoftwo
 \fi
}%
\providecommand \natexlab [1]{#1}%
\providecommand \enquote  [1]{``#1''}%
\providecommand \bibnamefont  [1]{#1}%
\providecommand \bibfnamefont [1]{#1}%
\providecommand \citenamefont [1]{#1}%
\providecommand \href@noop [0]{\@secondoftwo}%
\providecommand \href [0]{\begingroup \@sanitize@url \@href}%
\providecommand \@href[1]{\@@startlink{#1}\@@href}%
\providecommand \@@href[1]{\endgroup#1\@@endlink}%
\providecommand \@sanitize@url [0]{\catcode `\\12\catcode `\$12\catcode
  `\&12\catcode `\#12\catcode `\^12\catcode `\_12\catcode `\%12\relax}%
\providecommand \@@startlink[1]{}%
\providecommand \@@endlink[0]{}%
\providecommand \url  [0]{\begingroup\@sanitize@url \@url }%
\providecommand \@url [1]{\endgroup\@href {#1}{\urlprefix }}%
\providecommand \urlprefix  [0]{URL }%
\providecommand \Eprint [0]{\href }%
\providecommand \doibase [0]{https://doi.org/}%
\providecommand \selectlanguage [0]{\@gobble}%
\providecommand \bibinfo  [0]{\@secondoftwo}%
\providecommand \bibfield  [0]{\@secondoftwo}%
\providecommand \translation [1]{[#1]}%
\providecommand \BibitemOpen [0]{}%
\providecommand \bibitemStop [0]{}%
\providecommand \bibitemNoStop [0]{.\EOS\space}%
\providecommand \EOS [0]{\spacefactor3000\relax}%
\providecommand \BibitemShut  [1]{\csname bibitem#1\endcsname}%
\let\auto@bib@innerbib\@empty
\bibitem [{\citenamefont {Baumgratz}\ \emph {et~al.}(2014)\citenamefont
  {Baumgratz}, \citenamefont {Cramer},\ and\ \citenamefont
  {Plenio}}]{Baumgratz2014}%
  \BibitemOpen
  \bibfield  {author} {\bibinfo {author} {\bibfnamefont {T.}~\bibnamefont
  {Baumgratz}}, \bibinfo {author} {\bibfnamefont {M.}~\bibnamefont {Cramer}},\
  and\ \bibinfo {author} {\bibfnamefont {M.~B.}\ \bibnamefont {Plenio}},\
  }\bibfield  {title} {\bibinfo {title} {Quantifying coherence},\ }\href@noop
  {} {\bibfield  {journal} {\bibinfo  {journal} {Phys. Rev. Lett.}\ }\textbf
  {\bibinfo {volume} {113}},\ \bibinfo {pages} {140401} (\bibinfo {year}
  {2014})}\BibitemShut {NoStop}%
\bibitem [{\citenamefont {Streltsov}\ \emph {et~al.}(2017)\citenamefont
  {Streltsov}, \citenamefont {Adesso},\ and\ \citenamefont
  {Plenio}}]{Streltsov2017}%
  \BibitemOpen
  \bibfield  {author} {\bibinfo {author} {\bibfnamefont {A.}~\bibnamefont
  {Streltsov}}, \bibinfo {author} {\bibfnamefont {G.}~\bibnamefont {Adesso}},\
  and\ \bibinfo {author} {\bibfnamefont {M.~B.}\ \bibnamefont {Plenio}},\
  }\bibfield  {title} {\bibinfo {title} {Colloquium: Quantum coherence as a
  resource},\ }\href@noop {} {\bibfield  {journal} {\bibinfo  {journal} {Rev.
  Mod. Phys.}\ }\textbf {\bibinfo {volume} {89}},\ \bibinfo {pages} {041003}
  (\bibinfo {year} {2017})}\BibitemShut {NoStop}%
\bibitem [{\citenamefont {Winter}\ and\ \citenamefont
  {Yang}(2016)}]{WinterYang2016}%
  \BibitemOpen
  \bibfield  {author} {\bibinfo {author} {\bibfnamefont {A.}~\bibnamefont
  {Winter}}\ and\ \bibinfo {author} {\bibfnamefont {D.}~\bibnamefont {Yang}},\
  }\bibfield  {title} {\bibinfo {title} {Operational resource theory of
  coherence},\ }\href@noop {} {\bibfield  {journal} {\bibinfo  {journal} {Phys.
  Rev. Lett.}\ }\textbf {\bibinfo {volume} {116}},\ \bibinfo {pages} {120404}
  (\bibinfo {year} {2016})}\BibitemShut {NoStop}%
\bibitem [{\citenamefont {Napoli}\ \emph {et~al.}(2016)\citenamefont {Napoli},
  \citenamefont {Bromley}, \citenamefont {Cianciaruso}, \citenamefont {Piani},
  \citenamefont {Johnston},\ and\ \citenamefont {Adesso}}]{Napoli2016}%
  \BibitemOpen
  \bibfield  {author} {\bibinfo {author} {\bibfnamefont {C.}~\bibnamefont
  {Napoli}}, \bibinfo {author} {\bibfnamefont {T.~R.}\ \bibnamefont {Bromley}},
  \bibinfo {author} {\bibfnamefont {M.}~\bibnamefont {Cianciaruso}}, \bibinfo
  {author} {\bibfnamefont {M.}~\bibnamefont {Piani}}, \bibinfo {author}
  {\bibfnamefont {N.}~\bibnamefont {Johnston}},\ and\ \bibinfo {author}
  {\bibfnamefont {G.}~\bibnamefont {Adesso}},\ }\bibfield  {title} {\bibinfo
  {title} {Robustness of coherence: An operational and observable measure of
  quantum coherence},\ }\href@noop {} {\bibfield  {journal} {\bibinfo
  {journal} {Phys. Rev. Lett.}\ }\textbf {\bibinfo {volume} {116}},\ \bibinfo
  {pages} {150502} (\bibinfo {year} {2016})}\BibitemShut {NoStop}%
\bibitem [{\citenamefont {Lostaglio}\ \emph {et~al.}(2015)\citenamefont
  {Lostaglio}, \citenamefont {Korzekwa}, \citenamefont {Jennings},\ and\
  \citenamefont {Rudolph}}]{Lostaglio2015}%
  \BibitemOpen
  \bibfield  {author} {\bibinfo {author} {\bibfnamefont {M.}~\bibnamefont
  {Lostaglio}}, \bibinfo {author} {\bibfnamefont {K.}~\bibnamefont {Korzekwa}},
  \bibinfo {author} {\bibfnamefont {D.}~\bibnamefont {Jennings}},\ and\
  \bibinfo {author} {\bibfnamefont {T.}~\bibnamefont {Rudolph}},\ }\bibfield
  {title} {\bibinfo {title} {Quantum coherence, time-translation symmetry, and
  thermodynamics},\ }\href@noop {} {\bibfield  {journal} {\bibinfo  {journal}
  {Phys. Rev. X}\ }\textbf {\bibinfo {volume} {5}},\ \bibinfo {pages} {021001}
  (\bibinfo {year} {2015})}\BibitemShut {NoStop}%
\bibitem [{\citenamefont {Marvian}\ and\ \citenamefont
  {Spekkens}(2016)}]{MarvianSpekkens2016}%
  \BibitemOpen
  \bibfield  {author} {\bibinfo {author} {\bibfnamefont {I.}~\bibnamefont
  {Marvian}}\ and\ \bibinfo {author} {\bibfnamefont {R.~W.}\ \bibnamefont
  {Spekkens}},\ }\bibfield  {title} {\bibinfo {title} {How to quantify
  coherence: Distinguishing speakable and unspeakable notions},\ }\href@noop {}
  {\bibfield  {journal} {\bibinfo  {journal} {Phys. Rev. A}\ }\textbf {\bibinfo
  {volume} {94}},\ \bibinfo {pages} {052324} (\bibinfo {year}
  {2016})}\BibitemShut {NoStop}%
\bibitem [{\citenamefont {Mani}\ and\ \citenamefont
  {Karimipour}(2015)}]{ManiKarimipour2015}%
  \BibitemOpen
  \bibfield  {author} {\bibinfo {author} {\bibfnamefont {A.}~\bibnamefont
  {Mani}}\ and\ \bibinfo {author} {\bibfnamefont {V.}~\bibnamefont
  {Karimipour}},\ }\bibfield  {title} {\bibinfo {title} {Cohering and
  decohering power of quantum channels},\ }\href@noop {} {\bibfield  {journal}
  {\bibinfo  {journal} {Phys. Rev. A}\ }\textbf {\bibinfo {volume} {92}},\
  \bibinfo {pages} {032331} (\bibinfo {year} {2015})}\BibitemShut {NoStop}%
\bibitem [{\citenamefont {Zanardi}\ \emph
  {et~al.}(2017{\natexlab{a}})\citenamefont {Zanardi}, \citenamefont
  {Styliaris},\ and\ \citenamefont
  {Campos~Venuti}}]{ZanardiStyliarisVenuti2017}%
  \BibitemOpen
  \bibfield  {author} {\bibinfo {author} {\bibfnamefont {P.}~\bibnamefont
  {Zanardi}}, \bibinfo {author} {\bibfnamefont {G.}~\bibnamefont {Styliaris}},\
  and\ \bibinfo {author} {\bibfnamefont {L.}~\bibnamefont {Campos~Venuti}},\
  }\bibfield  {title} {\bibinfo {title} {Coherence-generating power of quantum
  unitary maps and beyond},\ }\href@noop {} {\bibfield  {journal} {\bibinfo
  {journal} {Phys. Rev. A}\ }\textbf {\bibinfo {volume} {95}},\ \bibinfo
  {pages} {052306} (\bibinfo {year} {2017}{\natexlab{a}})}\BibitemShut
  {NoStop}%
\bibitem [{\citenamefont {Zanardi}\ \emph
  {et~al.}(2017{\natexlab{b}})\citenamefont {Zanardi}, \citenamefont
  {Styliaris},\ and\ \citenamefont
  {Campos~Venuti}}]{ZanardiStyliarisVenuti2017b}%
  \BibitemOpen
  \bibfield  {author} {\bibinfo {author} {\bibfnamefont {P.}~\bibnamefont
  {Zanardi}}, \bibinfo {author} {\bibfnamefont {G.}~\bibnamefont {Styliaris}},\
  and\ \bibinfo {author} {\bibfnamefont {L.}~\bibnamefont {Campos~Venuti}},\
  }\bibfield  {title} {\bibinfo {title} {Measures of coherence-generating power
  for quantum unital operations},\ }\href@noop {} {\bibfield  {journal}
  {\bibinfo  {journal} {Phys. Rev. A}\ }\textbf {\bibinfo {volume} {95}},\
  \bibinfo {pages} {052307} (\bibinfo {year} {2017}{\natexlab{b}})}\BibitemShut
  {NoStop}%
\bibitem [{\citenamefont {Styliaris}\ \emph {et~al.}(2018)\citenamefont
  {Styliaris}, \citenamefont {Campos~Venuti},\ and\ \citenamefont
  {Zanardi}}]{StyliarisVenutiZanardi2018}%
  \BibitemOpen
  \bibfield  {author} {\bibinfo {author} {\bibfnamefont {G.}~\bibnamefont
  {Styliaris}}, \bibinfo {author} {\bibfnamefont {L.}~\bibnamefont
  {Campos~Venuti}},\ and\ \bibinfo {author} {\bibfnamefont {P.}~\bibnamefont
  {Zanardi}},\ }\bibfield  {title} {\bibinfo {title} {Coherence-generating
  power of quantum dephasing processes},\ }\href@noop {} {\bibfield  {journal}
  {\bibinfo  {journal} {Phys. Rev. A}\ }\textbf {\bibinfo {volume} {97}},\
  \bibinfo {pages} {032304} (\bibinfo {year} {2018})}\BibitemShut {NoStop}%
\bibitem [{\citenamefont {Zhang}\ \emph {et~al.}(2018)\citenamefont {Zhang},
  \citenamefont {Ma}, \citenamefont {Chen},\ and\ \citenamefont
  {Fei}}]{ZhangMaChenFei2018}%
  \BibitemOpen
  \bibfield  {author} {\bibinfo {author} {\bibfnamefont {L.}~\bibnamefont
  {Zhang}}, \bibinfo {author} {\bibfnamefont {Z.}~\bibnamefont {Ma}}, \bibinfo
  {author} {\bibfnamefont {Z.}~\bibnamefont {Chen}},\ and\ \bibinfo {author}
  {\bibfnamefont {S.-M.}\ \bibnamefont {Fei}},\ }\bibfield  {title} {\bibinfo
  {title} {Coherence generating power of unitary transformations via
  probabilistic average},\ }\href@noop {} {\bibfield  {journal} {\bibinfo
  {journal} {Quantum Inf. Process.}\ }\textbf {\bibinfo {volume} {17}},\
  \bibinfo {pages} {186} (\bibinfo {year} {2018})}\BibitemShut {NoStop}%
\bibitem [{\citenamefont {Cho}\ and\ \citenamefont {Bang}(2026)}]{ChoBang2026}%
  \BibitemOpen
  \bibfield  {author} {\bibinfo {author} {\bibfnamefont {K.}~\bibnamefont
  {Cho}}\ and\ \bibinfo {author} {\bibfnamefont {J.}~\bibnamefont {Bang}},\
  }\bibfield  {title} {\bibinfo {title} {Entangling power and its deviation: A
  quantitative analysis on input-state dependence and variability in
  entanglement generation},\ }\href@noop {} {\bibfield  {journal} {\bibinfo
  {journal} {Phys. Rev. A}\ }\textbf {\bibinfo {volume} {113}},\ \bibinfo
  {pages} {012442} (\bibinfo {year} {2026})}\BibitemShut {NoStop}%
\bibitem [{\citenamefont {Aubry}\ and\ \citenamefont
  {Andr{\'e}}(1980)}]{AubryAndre1980}%
  \BibitemOpen
  \bibfield  {author} {\bibinfo {author} {\bibfnamefont {S.}~\bibnamefont
  {Aubry}}\ and\ \bibinfo {author} {\bibfnamefont {G.}~\bibnamefont
  {Andr{\'e}}},\ }\bibfield  {title} {\bibinfo {title} {Analyticity breaking
  and {Anderson} localization in incommensurate lattices},\ }\href@noop {}
  {\bibfield  {journal} {\bibinfo  {journal} {Ann. Israel Phys. Soc.}\ }\textbf
  {\bibinfo {volume} {3}},\ \bibinfo {pages} {133} (\bibinfo {year}
  {1980})}\BibitemShut {NoStop}%
\bibitem [{\citenamefont {Biddle}\ \emph {et~al.}(2009)\citenamefont {Biddle},
  \citenamefont {Wang}, \citenamefont {Priour},\ and\ \citenamefont
  {Das~Sarma}}]{BiddleDasSarma2009}%
  \BibitemOpen
  \bibfield  {author} {\bibinfo {author} {\bibfnamefont {J.}~\bibnamefont
  {Biddle}}, \bibinfo {author} {\bibfnamefont {B.}~\bibnamefont {Wang}},
  \bibinfo {author} {\bibfnamefont {D.~J.}\ \bibnamefont {Priour},
  \bibfnamefont {Jr.}},\ and\ \bibinfo {author} {\bibfnamefont
  {S.}~\bibnamefont {Das~Sarma}},\ }\bibfield  {title} {\bibinfo {title}
  {Localization in one-dimensional incommensurate lattices beyond the
  {Aubry--Andr\'e} model},\ }\href@noop {} {\bibfield  {journal} {\bibinfo
  {journal} {Phys. Rev. A}\ }\textbf {\bibinfo {volume} {80}},\ \bibinfo
  {pages} {021603(R)} (\bibinfo {year} {2009})}\BibitemShut {NoStop}%
\bibitem [{\citenamefont {Lahini}\ \emph {et~al.}(2009)\citenamefont {Lahini},
  \citenamefont {Pugatch}, \citenamefont {Pozzi}, \citenamefont {Sorel},
  \citenamefont {Morandotti}, \citenamefont {Davidson},\ and\ \citenamefont
  {Silberberg}}]{Lahini2009}%
  \BibitemOpen
  \bibfield  {author} {\bibinfo {author} {\bibfnamefont {Y.}~\bibnamefont
  {Lahini}}, \bibinfo {author} {\bibfnamefont {R.}~\bibnamefont {Pugatch}},
  \bibinfo {author} {\bibfnamefont {F.}~\bibnamefont {Pozzi}}, \bibinfo
  {author} {\bibfnamefont {M.}~\bibnamefont {Sorel}}, \bibinfo {author}
  {\bibfnamefont {R.}~\bibnamefont {Morandotti}}, \bibinfo {author}
  {\bibfnamefont {N.}~\bibnamefont {Davidson}},\ and\ \bibinfo {author}
  {\bibfnamefont {Y.}~\bibnamefont {Silberberg}},\ }\bibfield  {title}
  {\bibinfo {title} {Observation of a localization transition in quasiperiodic
  photonic lattices},\ }\href@noop {} {\bibfield  {journal} {\bibinfo
  {journal} {Phys. Rev. Lett.}\ }\textbf {\bibinfo {volume} {103}},\ \bibinfo
  {pages} {013901} (\bibinfo {year} {2009})}\BibitemShut {NoStop}%
\bibitem [{\citenamefont {Rossi}\ \emph {et~al.}(2015)\citenamefont {Rossi},
  \citenamefont {Torsello},\ and\ \citenamefont {Hancock}}]{Rossi2015}%
  \BibitemOpen
  \bibfield  {author} {\bibinfo {author} {\bibfnamefont {L.}~\bibnamefont
  {Rossi}}, \bibinfo {author} {\bibfnamefont {A.}~\bibnamefont {Torsello}},\
  and\ \bibinfo {author} {\bibfnamefont {E.~R.}\ \bibnamefont {Hancock}},\
  }\bibfield  {title} {\bibinfo {title} {Measuring graph similarity through
  continuous-time quantum walks and the quantum {Jensen--Shannon} divergence},\
  }\href@noop {} {\bibfield  {journal} {\bibinfo  {journal} {Phys. Rev. E}\
  }\textbf {\bibinfo {volume} {91}},\ \bibinfo {pages} {022815} (\bibinfo
  {year} {2015})}\BibitemShut {NoStop}%
\bibitem [{\citenamefont {Henry}\ \emph {et~al.}(2021)\citenamefont {Henry},
  \citenamefont {Thabet}, \citenamefont {Dalyac},\ and\ \citenamefont
  {Henriet}}]{Henry2021}%
  \BibitemOpen
  \bibfield  {author} {\bibinfo {author} {\bibfnamefont {L.-P.}\ \bibnamefont
  {Henry}}, \bibinfo {author} {\bibfnamefont {S.}~\bibnamefont {Thabet}},
  \bibinfo {author} {\bibfnamefont {C.}~\bibnamefont {Dalyac}},\ and\ \bibinfo
  {author} {\bibfnamefont {L.}~\bibnamefont {Henriet}},\ }\bibfield  {title}
  {\bibinfo {title} {Quantum evolution kernel: Machine learning on graphs with
  programmable arrays of qubits},\ }\href@noop {} {\bibfield  {journal}
  {\bibinfo  {journal} {Phys. Rev. A}\ }\textbf {\bibinfo {volume} {104}},\
  \bibinfo {pages} {032416} (\bibinfo {year} {2021})}\BibitemShut {NoStop}%
\bibitem [{\citenamefont {Albrecht}\ \emph {et~al.}(2023)\citenamefont
  {Albrecht}, \citenamefont {Dalyac}, \citenamefont {Leclerc}, \citenamefont
  {Ortiz-Guti{\'e}rrez}, \citenamefont {Thabet} \emph {et~al.}}]{Albrecht2023}%
  \BibitemOpen
  \bibfield  {author} {\bibinfo {author} {\bibfnamefont {B.}~\bibnamefont
  {Albrecht}}, \bibinfo {author} {\bibfnamefont {C.}~\bibnamefont {Dalyac}},
  \bibinfo {author} {\bibfnamefont {L.}~\bibnamefont {Leclerc}}, \bibinfo
  {author} {\bibfnamefont {L.}~\bibnamefont {Ortiz-Guti{\'e}rrez}}, \bibinfo
  {author} {\bibfnamefont {S.}~\bibnamefont {Thabet}}, \emph {et~al.},\
  }\bibfield  {title} {\bibinfo {title} {Quantum feature maps for graph machine
  learning on a neutral atom quantum processor},\ }\href@noop {} {\bibfield
  {journal} {\bibinfo  {journal} {Phys. Rev. A}\ }\textbf {\bibinfo {volume}
  {107}},\ \bibinfo {pages} {042615} (\bibinfo {year} {2023})}\BibitemShut
  {NoStop}%
\bibitem [{\citenamefont {Havl{\'i}{\v c}ek}\ \emph {et~al.}(2019)\citenamefont
  {Havl{\'i}{\v c}ek}, \citenamefont {C{\'o}rcoles}, \citenamefont {Temme},
  \citenamefont {Harrow}, \citenamefont {Kandala}, \citenamefont {Chow},\ and\
  \citenamefont {Gambetta}}]{Havlicek2019}%
  \BibitemOpen
  \bibfield  {author} {\bibinfo {author} {\bibfnamefont {V.}~\bibnamefont
  {Havl{\'i}{\v c}ek}}, \bibinfo {author} {\bibfnamefont {A.~D.}\ \bibnamefont
  {C{\'o}rcoles}}, \bibinfo {author} {\bibfnamefont {K.}~\bibnamefont {Temme}},
  \bibinfo {author} {\bibfnamefont {A.~W.}\ \bibnamefont {Harrow}}, \bibinfo
  {author} {\bibfnamefont {A.}~\bibnamefont {Kandala}}, \bibinfo {author}
  {\bibfnamefont {J.~M.}\ \bibnamefont {Chow}},\ and\ \bibinfo {author}
  {\bibfnamefont {J.~M.}\ \bibnamefont {Gambetta}},\ }\bibfield  {title}
  {\bibinfo {title} {Supervised learning with quantum-enhanced feature
  spaces},\ }\href@noop {} {\bibfield  {journal} {\bibinfo  {journal} {Nature}\
  }\textbf {\bibinfo {volume} {567}},\ \bibinfo {pages} {209} (\bibinfo {year}
  {2019})}\BibitemShut {NoStop}%
\bibitem [{\citenamefont {Schuld}\ and\ \citenamefont
  {Killoran}(2019)}]{SchuldKilloran2019}%
  \BibitemOpen
  \bibfield  {author} {\bibinfo {author} {\bibfnamefont {M.}~\bibnamefont
  {Schuld}}\ and\ \bibinfo {author} {\bibfnamefont {N.}~\bibnamefont
  {Killoran}},\ }\bibfield  {title} {\bibinfo {title} {Quantum machine learning
  in feature {Hilbert} spaces},\ }\href@noop {} {\bibfield  {journal} {\bibinfo
   {journal} {Phys. Rev. Lett.}\ }\textbf {\bibinfo {volume} {122}},\ \bibinfo
  {pages} {040504} (\bibinfo {year} {2019})}\BibitemShut {NoStop}%
\bibitem [{\citenamefont {{\.{Z}}yczkowski}\ and\ \citenamefont
  {Sommers}(2001)}]{ZyczkowskiSommers2001}%
  \BibitemOpen
  \bibfield  {author} {\bibinfo {author} {\bibfnamefont {K.}~\bibnamefont
  {{\.{Z}}yczkowski}}\ and\ \bibinfo {author} {\bibfnamefont {H.-J.}\
  \bibnamefont {Sommers}},\ }\bibfield  {title} {\bibinfo {title} {Induced
  measures in the space of mixed quantum states},\ }\href@noop {} {\bibfield
  {journal} {\bibinfo  {journal} {J. Phys. A: Math. Gen.}\ }\textbf {\bibinfo
  {volume} {34}},\ \bibinfo {pages} {7111} (\bibinfo {year}
  {2001})}\BibitemShut {NoStop}%
\bibitem [{\citenamefont {Collins}\ and\ \citenamefont
  {{\'S}niady}(2006)}]{CollinsSniady2006}%
  \BibitemOpen
  \bibfield  {author} {\bibinfo {author} {\bibfnamefont {B.}~\bibnamefont
  {Collins}}\ and\ \bibinfo {author} {\bibfnamefont {P.}~\bibnamefont
  {{\'S}niady}},\ }\bibfield  {title} {\bibinfo {title} {Integration with
  respect to the {Haar} measure on unitary, orthogonal and symplectic group},\
  }\href@noop {} {\bibfield  {journal} {\bibinfo  {journal} {Commun. Math.
  Phys.}\ }\textbf {\bibinfo {volume} {264}},\ \bibinfo {pages} {773} (\bibinfo
  {year} {2006})}\BibitemShut {NoStop}%
\bibitem [{\citenamefont {Feller}(1971)}]{Feller1971}%
  \BibitemOpen
  \bibfield  {author} {\bibinfo {author} {\bibfnamefont {W.}~\bibnamefont
  {Feller}},\ }\href@noop {} {\emph {\bibinfo {title} {An Introduction to
  Probability Theory and Its Applications, Vol. 2}}},\ \bibinfo {edition}
  {2nd}\ ed.\ (\bibinfo  {publisher} {John Wiley \& Sons},\ \bibinfo {address}
  {New York},\ \bibinfo {year} {1971})\BibitemShut {NoStop}%
\bibitem [{\citenamefont {Elben}\ \emph {et~al.}(2018)\citenamefont {Elben},
  \citenamefont {Vermersch}, \citenamefont {Dalmonte}, \citenamefont {Cirac},\
  and\ \citenamefont {Zoller}}]{Elben2018}%
  \BibitemOpen
  \bibfield  {author} {\bibinfo {author} {\bibfnamefont {A.}~\bibnamefont
  {Elben}}, \bibinfo {author} {\bibfnamefont {B.}~\bibnamefont {Vermersch}},
  \bibinfo {author} {\bibfnamefont {M.}~\bibnamefont {Dalmonte}}, \bibinfo
  {author} {\bibfnamefont {J.~I.}\ \bibnamefont {Cirac}},\ and\ \bibinfo
  {author} {\bibfnamefont {P.}~\bibnamefont {Zoller}},\ }\bibfield  {title}
  {\bibinfo {title} {{R\'enyi} entropies from random quenches in atomic
  {Hubbard} and spin models},\ }\href@noop {} {\bibfield  {journal} {\bibinfo
  {journal} {Phys. Rev. Lett.}\ }\textbf {\bibinfo {volume} {120}},\ \bibinfo
  {pages} {050406} (\bibinfo {year} {2018})}\BibitemShut {NoStop}%
\bibitem [{\citenamefont {Brydges}\ \emph {et~al.}(2019)\citenamefont
  {Brydges}, \citenamefont {Elben}, \citenamefont {Jurcevic}, \citenamefont
  {Vermersch}, \citenamefont {Maier}, \citenamefont {Lanyon}, \citenamefont
  {Zoller}, \citenamefont {Blatt},\ and\ \citenamefont {Roos}}]{Brydges2019}%
  \BibitemOpen
  \bibfield  {author} {\bibinfo {author} {\bibfnamefont {T.}~\bibnamefont
  {Brydges}}, \bibinfo {author} {\bibfnamefont {A.}~\bibnamefont {Elben}},
  \bibinfo {author} {\bibfnamefont {P.}~\bibnamefont {Jurcevic}}, \bibinfo
  {author} {\bibfnamefont {B.}~\bibnamefont {Vermersch}}, \bibinfo {author}
  {\bibfnamefont {C.}~\bibnamefont {Maier}}, \bibinfo {author} {\bibfnamefont
  {B.~P.}\ \bibnamefont {Lanyon}}, \bibinfo {author} {\bibfnamefont
  {P.}~\bibnamefont {Zoller}}, \bibinfo {author} {\bibfnamefont
  {R.}~\bibnamefont {Blatt}},\ and\ \bibinfo {author} {\bibfnamefont {C.~F.}\
  \bibnamefont {Roos}},\ }\bibfield  {title} {\bibinfo {title} {Probing
  {R\'enyi} entanglement entropy via randomized measurements},\ }\href@noop {}
  {\bibfield  {journal} {\bibinfo  {journal} {Science}\ }\textbf {\bibinfo
  {volume} {364}},\ \bibinfo {pages} {260} (\bibinfo {year}
  {2019})}\BibitemShut {NoStop}%
\bibitem [{\citenamefont {Nielsen}\ and\ \citenamefont
  {Chuang}(2010)}]{NielsenChuang2010}%
  \BibitemOpen
  \bibfield  {author} {\bibinfo {author} {\bibfnamefont {M.~A.}\ \bibnamefont
  {Nielsen}}\ and\ \bibinfo {author} {\bibfnamefont {I.~L.}\ \bibnamefont
  {Chuang}},\ }\href@noop {} {\emph {\bibinfo {title} {Quantum Computation and
  Quantum Information}}},\ \bibinfo {edition} {10th}\ ed.\ (\bibinfo
  {publisher} {Cambridge University Press},\ \bibinfo {address} {Cambridge},\
  \bibinfo {year} {2010})\BibitemShut {NoStop}%
\bibitem [{\citenamefont {Farhi}\ and\ \citenamefont
  {Gutmann}(1998)}]{FarhiGutmann1998}%
  \BibitemOpen
  \bibfield  {author} {\bibinfo {author} {\bibfnamefont {E.}~\bibnamefont
  {Farhi}}\ and\ \bibinfo {author} {\bibfnamefont {S.}~\bibnamefont
  {Gutmann}},\ }\bibfield  {title} {\bibinfo {title} {Quantum computation and
  decision trees},\ }\href@noop {} {\bibfield  {journal} {\bibinfo  {journal}
  {Phys. Rev. A}\ }\textbf {\bibinfo {volume} {58}},\ \bibinfo {pages} {915}
  (\bibinfo {year} {1998})}\BibitemShut {NoStop}%
\bibitem [{\citenamefont {Childs}(2009)}]{Childs2009}%
  \BibitemOpen
  \bibfield  {author} {\bibinfo {author} {\bibfnamefont {A.~M.}\ \bibnamefont
  {Childs}},\ }\bibfield  {title} {\bibinfo {title} {Universal computation by
  quantum walk},\ }\href@noop {} {\bibfield  {journal} {\bibinfo  {journal}
  {Phys. Rev. Lett.}\ }\textbf {\bibinfo {volume} {102}},\ \bibinfo {pages}
  {180501} (\bibinfo {year} {2009})}\BibitemShut {NoStop}%
\bibitem [{\citenamefont {Venegas-Andraca}(2012)}]{VenegasAndraca2012}%
  \BibitemOpen
  \bibfield  {author} {\bibinfo {author} {\bibfnamefont {S.~E.}\ \bibnamefont
  {Venegas-Andraca}},\ }\bibfield  {title} {\bibinfo {title} {Quantum walks: a
  comprehensive review},\ }\href@noop {} {\bibfield  {journal} {\bibinfo
  {journal} {Quantum Inf. Process.}\ }\textbf {\bibinfo {volume} {11}},\
  \bibinfo {pages} {1015} (\bibinfo {year} {2012})}\BibitemShut {NoStop}%
\bibitem [{\citenamefont {Chen}\ \emph {et~al.}(2024)\citenamefont {Chen},
  \citenamefont {Li},\ and\ \citenamefont {Li}}]{ChenLiLi2024}%
  \BibitemOpen
  \bibfield  {author} {\bibinfo {author} {\bibfnamefont {Z.}~\bibnamefont
  {Chen}}, \bibinfo {author} {\bibfnamefont {G.}~\bibnamefont {Li}},\ and\
  \bibinfo {author} {\bibfnamefont {L.}~\bibnamefont {Li}},\ }\bibfield
  {title} {\bibinfo {title} {Implementation of a continuous-time quantum walk
  on a sparse graph},\ }\href@noop {} {\bibfield  {journal} {\bibinfo
  {journal} {Phys. Rev. A}\ }\textbf {\bibinfo {volume} {110}},\ \bibinfo
  {pages} {052215} (\bibinfo {year} {2024})}\BibitemShut {NoStop}%
\bibitem [{\citenamefont {Baker}\ \emph {et~al.}(2024)\citenamefont {Baker},
  \citenamefont {Park}, \citenamefont {Yu}, \citenamefont {Ghukasyan},
  \citenamefont {Goktas},\ and\ \citenamefont {Radha}}]{Baker2023}%
  \BibitemOpen
  \bibfield  {author} {\bibinfo {author} {\bibfnamefont {J.~S.}\ \bibnamefont
  {Baker}}, \bibinfo {author} {\bibfnamefont {G.}~\bibnamefont {Park}},
  \bibinfo {author} {\bibfnamefont {K.}~\bibnamefont {Yu}}, \bibinfo {author}
  {\bibfnamefont {A.}~\bibnamefont {Ghukasyan}}, \bibinfo {author}
  {\bibfnamefont {O.}~\bibnamefont {Goktas}},\ and\ \bibinfo {author}
  {\bibfnamefont {S.~K.}\ \bibnamefont {Radha}},\ }\bibfield  {title} {\bibinfo
  {title} {Parallel hybrid quantum-classical machine learning for kernelized
  time-series classification},\ }\href@noop {} {\bibfield  {journal} {\bibinfo
  {journal} {Quantum Machine Intelligence}\ }\textbf {\bibinfo {volume} {6}},\
  \bibinfo {pages} {18} (\bibinfo {year} {2024})}\BibitemShut {NoStop}%
\end{thebibliography}

%

\end{document}